\documentclass[11pt]{article}  
\usepackage[fontsize=11.5pt]{scrextend}
\usepackage[text={6.5in,9in}]{geometry}
\usepackage{fullpage,  enumerate, booktabs}
\usepackage[compact]{titlesec}
\usepackage{amsmath, enumitem, bigints}
\usepackage{algorithm, algpseudocode}
\usepackage{graphicx}
\usepackage{wrapfig}
\usepackage{sidecap}
\usepackage[normalem]{ulem} 
\usepackage{multirow}
\usepackage[utf8]{inputenc}
\usepackage{threeparttable}
\usepackage[mathscr]{eucal}
\usepackage{geometry}
\usepackage{times}
\usepackage{courier}
\usepackage{float}
\usepackage{mathtools}
\usepackage{color}
\usepackage{subcaption}
\usepackage{bbm}
\usepackage{array}
\usepackage{booktabs}
\usepackage{multirow}
\usepackage{siunitx, tabularx}
\usepackage{adjustbox}
\usepackage{xr}
\usepackage{titling}
\usepackage{arydshln,,leftidx}
\RequirePackage[colorlinks,citecolor=blue,urlcolor=blue]{hyperref}
\usepackage{verbatim}
\usepackage[font=small,labelfont=bf]{caption}
\usepackage{ upgreek }
\usepackage[round]{natbib} 
\usepackage{color, url}
\usepackage{dsfont}
\usepackage{amssymb}
\usepackage{epstopdf}
\usepackage{bm}
\usepackage{setspace}

\usepackage[compact]{titlesec}
\titlespacing{\section}{0pt}{*0.3}{*0.3}
\titlespacing{\subsection}{0pt}{*0.3}{*0.3}
\titlespacing{\subsubsection}{0pt}{*0.1}{*0.1}

\def\bc{\begin{center}}   % Begin Enumerate
\def\ec{\end{center}}     % End Enumerate

\newtheorem{proposition}{Proposition}
\newtheorem{theorem}{Theorem}[section]
\newtheorem{lemma}[theorem]{Lemma}

\newenvironment{proof}{\paragraph{Proof:}}{\hfill$\square$}

\newcommand{\h}[1]{\widehat{#1}} 
\newcommand{\Stacked}[1]{\mathbf{#1}}
\newcommand{\StackedSymbol}[1]{\ensuremath{\boldsymbol{#1}}}
\newcommand{\localvar}[1]{\widetilde{#1}} 

\newcommand{\Mb}{\boldsymbol{\widehat{\beta}}^{\text{\;MLE}}}
\newcommand{\InvFI}{\boldsymbol{\widehat{\mathsf{I}}}^{\; -1}}

\newcommand{\grad}{{\nabla}}

\def\argmin{\mathop{\rm argmin}\limits}

\newcommand{\bi}{\begin{itemize}}
\newcommand{\ei}{\end{itemize}}
\newcommand{\bqa}{\begin{eqnarray*}}
\newcommand{\eqa}{\end{eqnarray*}}
\newcommand{\bqan}{\begin{eqnarray}}
\newcommand{\eqan}{\end{eqnarray}}

\newcommand{\Dg}{\text{Diag}}

\usepackage{natbib}

\newcommand{\twofigs}[2]{
\hbox to\hsize{\hss
\vbox{\psfig{figure=#1,width=2.7in,height=2.0in}}\qquad
\vbox{\psfig{figure=#2,width=2.7in,height=2.0in}}
\hss}}

\begin{document}

  \title{\bf Selective Inference for Sparse Multitask Regression with Applications in Neuroimaging}
  \author{Snigdha Panigrahi$^{1}$, Natasha Stewart$^{1}$, Chandra Sripada$^{2,3}$, and Elizaveta Levina$^{1}$ \\  $^{1}$Department of Statistics, University of Michigan\\ $^{2}$Department of Psychiatry, University of Michigan\\ $^{3}$Department of  Philosophy, University of Michigan}
    \date{}

\maketitle

Multi-task learning is frequently used  to model a set of related response variables from the same set of features, improving predictive performance and modeling accuracy relative to methods that handle each response variable separately. 
Despite the potential of multi-task learning to yield more powerful inference than single-task alternatives, prior work in this area has largely omitted uncertainty quantification.
Our focus in this paper is a common multi-task problem in neuroimaging, where the goal is to understand the relationship between multiple cognitive task scores (or other subject-level assessments) and brain connectome data collected from imaging. 
We propose a framework for selective inference to address this problem, with the flexibility to: (i) jointly identify the relevant covariates for each task through a sparsity-inducing penalty, and (ii) conduct valid inference in a model based on the estimated sparsity structure. 
Our framework offers a new conditional procedure for inference, based on a refinement of the selection event that yields a tractable selection-adjusted likelihood. This gives an approximate system of estimating equations for maximum likelihood inference, solvable via a single convex optimization problem, and enables us to efficiently form confidence intervals with approximately the correct coverage. Applied to both simulated data and data from the Adolescent Brain Cognitive Development (ABCD) study, our selective inference methods yield tighter confidence intervals than commonly used alternatives, such as data splitting. We also demonstrate through simulations that multi-task learning with selective inference can more accurately recover true signals than single-task methods.


\pagebreak

\section{Introduction}

Humans exhibit a diversity of cognitive abilities, which can be categorized as either fluid or crystallized. Fluid abilities are rooted in problem solving and manipulation of information, independent of prior learning \citep{gray2003neural,blair2006similar}. Examples include the ability to store and manipulate items in short-term memory (working memory) and detect subtle patterns in sequences (matrix reasoning). Crystallized abilities are rooted in determinate facts that are trained through prior learning \citep{cattell1943measurement, horn1997human}, including reading comprehension and vocabulary. According to proponents of a general ability model, although fluid tasks are superficially quite different from crystallized tasks, there is an underlying shared ability that drives performance on both kinds of tasks \citep{spearman1961general, humphreys1979construct, carroll1993human}. The question of whether cognitive abilities are better understood in terms of a single, general ability or separate fluid and crystallized abilities has been debated for over 100 years without a firm consensus emerging. 

Functional neuroimaging offers new opportunities to non-invasively investigate the neurological organization of cognitive abilities, avoiding the need to rely solely on behavioral data. Recent efforts have focused on using neuroimaging data to predict performance across cognitive tasks and enhance cartographic understanding by linking cognitive functions to particular brain networks. Cartographic maps for various cognitive abilities could add evidence for the general ability hypothesis if they show spatial overlap. Alternatively, these maps might show that individual cognitive abilities exhibit spatially distinct patterns in the brain, strengthening the case for neurally separate abilities. 

In this paper, we leverage multi-task learning (MTL) to advance both predictive and cartographic goals in the study of cognitive abilities. MTL is an important tool for modeling related response variables, such as performance on cognitive tasks,  that have common predictors and potentially common patterns of dependence between the predictors and responses. Algorithms that use MTL are known to improve predictive accuracy by accounting for shared information between the tasks. Many different MTL algorithms have been developed, including regression methods that impose inter-task dependence through shared sparsity or low-rank constraints on the multi-task regression coefficients; see \cite{zhang2021survey} and references therein for a thorough survey on the topic. We utilize a regression-based MTL method to identify a potential common neurological basis for fluid and crystallized abilities.

In addition to predicting performance across different cognitive tasks and visualizing the involvement of different brain regions, we also address an important question that has been neglected: statistical inference for the selected neurological models. We introduce a novel, two-step procedure that can recover shared signals in neuroimaging data and subsequently quantify uncertainty via selective inference. The first step in our procedure entails using a randomized MTL algorithm with a sparsity-inducing penalty to jointly identify a model for each of the cognitive tasks. In stage two, we conduct selective inference in the model chosen from the estimated sparsity structure, enabling us to test the significance of the shared signals recovered in the first step. This hypothesis testing framework, which offers the flexibility of choosing the model under which to conduct inference {\em after} examining the data, has the potential to significantly expand our understanding of the relationship between cognition and neurological features extracted from imaging data.

\subsection{Contributions to the Adolescent Brain Cognitive Development (ABCD) Study} 
The behavioral and brain data we study come from the Adolescent Brain Cognitive Development (ABCD) study, one of the most extensive efforts to track the brain development of a large cohort of children in the United States, with close to 12,000 children enrolled across 23 research sites. The ABCD study includes an 11-task neurocognitive battery \citep{luciana2018adolescent}, primarily based on the NIH Toolbox for the Assessment of Neurological and Behavioral Function, with several additional tasks to ensure comprehensive coverage across cognitive domains. Importantly, the ABCD battery includes two classic crystallized ability tasks (picture vocabulary and reading comprehension) as well as two classic fluid ability tasks (list working memory and matrix reasoning). It also uses multiple imaging modalities to track the participants' brain development \citep{hagler2019image}. Our focus is resting-state fMRI, which captures patterns of spontaneous activation throughout the brain while subjects are at rest in the scanner, yielding maps of functional connectivity. Previous studies have investigated the role of functional connectivity in general cognitive ability \citep{finn2015functional, hearne2016functional, sripada2021brain-wide, anderson2022investigating, tong2022transdiagnostic} as well as specific abilities like working memory \citep{markett2018working} and matrix reasoning \citep{fraenz2021interindividual}. 

Few previous studies have jointly modeled performance across cognitive domains to improve predictive performance or to better understand the spatial relationships between tasks. One notable exception is  \cite{adeli2019multi}, who used MTL to jointly predict the development of general and specific cognitive abilities in infants and young children over time. Other studies have implicitly leveraged joint model selection to predict performance on individual cognitive domains using connections associated with general ability \citep{anderson2022investigating, tong2022transdiagnostic}. In contrast to these previous studies, we do not use any measure of general ability computed from behavioral data. Our method relies on MTL to uncover the relationships between tasks from brain-behavior patterns, providing a new framework to assess the shared and distinct neural contributions to fluid and crystallized intelligence. 

A further limitation of the studies mentioned above, as well as most previous methodological developments for MTL, is an emphasis on prediction over inference. Prediction is typically not the final goal in cognitive neuroscience applications, however, since directly measuring cognitive ability from behavioral data is easier and less expensive than acquiring brain imaging data. Furthermore, predictive accuracy has limited explanatory value by itself since very different models can sometimes achieve similar predictive performance.  A combination of prediction and statistical inference could shed light on the neurobiological basis of cognition, identifying brain-behavior relationships and measuring both the strength and the degree of confidence for each relationship.

\subsection{Novel Contributions to Selective Inference} Selective inference has been studied extensively for single-task prediction. A common approach to selective inference, described by \cite{fithian2017optimal}, accounts for bias from model selection by conditioning on the chosen model. The key to this approach is characterizing the selection event in a sufficiently simple form. For models chosen based on the LASSO, \cite{exact_lasso} developed the influential polyhedral method, reducing the selection event to a series of linear inequalities in the response variable. Many different model selection events for a Gaussian response variable have subsequently been identified with a similar set of inequalities, yielding conditional distributions that can be used for inference  \citep[][among others]{suzumura2017selective, liu2018more, taylor2018post, zhao2019selective, tanizaki2020computing}. 

Despite the convenience of the polyhedral method, the resulting confidence intervals can have infinite expected length for Gaussian regression \citep{kivaranovic2018expected}. The loss in inferential power can be remedied by applying conditional inference to a randomized problem, for example, by adding random noise to the response variable \citep{randomized_response} or by holding out some samples during selection, known as data carving \citep{fithian2017optimal, panigrahi2019carving}. Under randomized versions of single-task algorithms such as the LASSO, \cite{randomized_response} obtained a pivot for each selected parameter after eliminating other (nuisance) parameters. More recently, \cite{panigrahi2019approximate} and \cite{panigrahi2018scalable, selective_bayesian, panigrahi2022integrativeradio} build on the polyhedral method to introduce a tractable likelihood that allows for both frequentist and Bayesian selective inference, using a prior in conjunction with the likelihood for the latter. 

Unfortunately, the  polyhedral methods for single-task algorithms do not generalize well to the multi-task setting because the usual conditioning event for single-task methods does not have a simple characterization under joint model selection. Remarkably, there is a proper subset of the usual conditioning event that makes selective inference feasible for our particular MTL procedure without sacrificing much power, leading to an easy-to-solve system of estimating equations. To the best of our knowledge, this is the first selective inference method for the multi-task setting.  

The remaining paper is organized as follows. Section \ref{sec:method} presents our algorithm for estimating the shared sparsity structure.
Section \ref{sec:MLE} develops our method for MLE-based selective inference.  
In Section \ref{sec:ABCD}, we apply our methods to the ABCD study data for identifying the neurobiological underpinnings of fluid and crystallized intelligence. 
A brief discussion in Section \ref{discussion} concludes our paper.

\section{Multi-Task Learning for Joint Model Selection}
\label{sec:method}

\subsection{The Two-Stage Model Selection and Inference Protocol}
\label{sec:two_stage}

For ease of presentation, we first review the two-stage protocol for model selection and subsequent selective inference for the (single-task) randomized LASSO described in \cite{panigrahi2019approximate}.  %, before proceeding to the multi-task case. 
Suppose we observe a predictor matrix 
$X \in \mathbb{R}^{n \times p}$
with fixed entries,  a random response vector  $y \sim N(\mu, \Sigma) \in \mathbb{R}^n$ with $\mu$ unknown and $\Sigma$ known, and an independent randomization variable $\omega \sim N(0,\Omega) \in \mathbb{R}^p$ with known $\Omega$. 
In the first stage, we identify a sparse linear model through a randomized LASSO regression.
We estimate the regression coefficients, $\Theta \in \mathbb{R}^p$, by solving
\begin{equation}
    \h{\Theta} = \argmin_{\Theta} \mathcal{L}\left(\Theta; y, X, \omega \right) + \| \Lambda \Theta \|_1,
\end{equation}
where $\Lambda \in \mathbb{R}^{p \times p}$ represents a diagonal matrix of feature-specific tuning parameters. 
The randomized loss $\mathcal{L} \left( \cdot; y, X, \omega\right)$ is given by
\begin{equation}
\label{Randomized Loss}
\mathcal{L}\left(\Theta; y, X, \omega \right) = \frac{1}{2} || y - X \Theta ||_2^2  - {\omega'} \Theta + \frac{\epsilon}{2} \|\Theta\|_2^2.
\end{equation}  We use $M'$ to denote the transpose of a matrix $M$. 
The ridge term in the loss function, with a small $\epsilon$, is included to ensure the existence of an optimal solution; see \cite{harris2016selective}.  The estimated support set, $\h{E}= \text{Supp}(\h{\Theta})$, is treated as random, reflecting the fact that different samples of $y$ and $\omega$ will yield different active sets. Suppose we observe $\h{E}(y,\omega)=E$ on our specific data. We use $X_{E} \in \mathbb{R}^{n \times |E|}$ to represent the restriction of $X$ to the columns indexed by the set $E$. 

After observing the active set $E$, we assume a linear model of the form $y  \sim N(X_{E}\beta_{E}, \sigma^2 \rm{I_n})$, where $\sigma$ is fixed and $\rm{I_n}$ is the $n \times n$ identity matrix. This is a simple way to specify a model based on the active set, but selective inference affords us the flexibility to choose from a range of   specifications. For instance, the two-stage protocol we describe here can also be used for other linear models, such as those involving linear combinations of basis functions constructed from the active features. 
%The procedure could also be easily adapted to accommodate generalizations of ordinary linear regression, such as generalized linear models.  

In the second stage, we aim to construct \mbox{$100(1-\alpha)\%$} confidence intervals for the best linear coefficients,
\begin{equation}
\label{best:linear:pred}
\beta_{E} = \argmin_{\beta \in \mathbb{R}^{|E|}} \  \mathbb{E} \left[ \|y - X_{E} \beta \|_2^2 \right], 
\end{equation}
under the selected model. 
We note that some recent post-selection work by \cite{hong2018overfitting, wang2020debiased, panigrahi2022treatment} has focused on overfitting bias from model selection.
However, the parameters in \eqref{best:linear:pred} are well-defined regardless of model misspecification, which could be either due to underfitting or overfitting bias.

A conditional likelihood for $\beta_{E}$ is obtained by conditioning upon the partition
$$ \mathcal{P}_{E}= \{y \in \mathbb{R}^n, \omega\in \mathbb{R}^p: \h{E}(y, \omega)=E\},$$ 
that contains all instances leading us to observe the estimated support set $E$. 
Letting $\rho\left(x ;\mu,\Omega\right)$ be the multivariate normal density function with mean vector $\mu$ and covariance $\Omega$, we can write this likelihood as 
\begin{equation}
    \begin{aligned}
    y \lvert \; \h{E}=E & \propto \dfrac{\rho\left(y;X_{E}\beta_{E}, \sigma^2\rm{I_n}\right)}{\bigintsss_{\mathcal{P}_{E}}\rho\left(\localvar{y};X_{E}\beta_{E}, \sigma^2\rm{I_n}\right) \cdot   \rho(\localvar{w}; 0, \Omega) d \localvar{w} d \localvar{y}} .
    \end{aligned}
    \label{cond:lik}
\end{equation}
The likelihood in \eqref{cond:lik} does not have a closed form since the normalizing constant is intractable.
Instead, a tractable version of the likelihood function, called the selection-adjusted likelihood, is obtained by conditioning on a proper subset of the actual selection event:
$$\left\{\h{S}(y, \omega)=S\right\} \subseteq \left\{\h{E}(y, \omega)=E\right\}.$$
The refined conditioning event, on the left-hand side of the previous display, can be characterized by the polyhedral partition
$$\{ y \in \mathbb{R}^n, \omega \in \mathbb{R}^p: A y + B\omega \leq c \}$$ for fixed matrices $A$, $B$ and $c$. 
For example, \cite{exact_lasso} identify a polyhedral partition by conditioning further on the active signs of the LASSO solution, alongside the estimated support set $E$.

Approximate inference is then obtained by centering interval estimates around the maximum likelihood estimate (MLE) of the selection-adjusted likelihood,  $\h{\beta}_{E}^{{\text{\; MLE}}}$, and using the observed Fisher information matrix, $\h{I}(\h{\beta}_{E}^{{{\text{\; MLE}}}})$, to estimate the variance.    
This yields post-selection confidence intervals for $\beta_{E}$ of the form 
$$\h{\beta}_{E,j}^{{\text{\; MLE}}} \pm z_{1-\alpha/2}  {\sqrt{ \h{I}_{jj}^{\; -1}(\h{\beta}_{E}^{{{\text{\; MLE}}}}) }}, \; j\in E;$$
where $v_j$ is the $j^{\text{th}}$ entry of the vector $v$, $M_{ij}$ is the $(i,j)^{\text{th}}$ element of a matrix $M$, and $z_{q}$ is the $q^{\text{th}}$ quantile of the standard normal distribution.

\subsection{An Objective Function with Shared Sparsity for MTL}
\label{sec:2.1}

We next set up the objective function in the multi-task setting.  Suppose we have $K$ regression tasks, with $K$ distinct response variables and a common set of $p$ predictors. For $k\in [K]$,  let  $n_k$ denote the sample size available for task $k$, $y^{(k)} \in \mathbb{R}^{n_k}$ denote the response vector for the $k^{\text{th}}$ task, and $X^{(k)}\in \mathbb{R}^{n_k \times p}$ denote the corresponding predictor matrix.   We assume that the predictors in each task have been centered and do not include intercept terms in the regression.
Consistent with the two-stage procedure described in Section \ref{sec:two_stage}, we introduce a randomization variable for each task.
Let $\omega^{(k)}\in  \mathbb{R}^{p}$ denote a Gaussian randomization variable such that (i) $\omega^{(k)} \sim \text{N} \left(0, \Omega^{(k)} \right)$ for $k\in [K]$; (ii) $\omega^{(k)}$ is independent of $\omega^{(k')}$ for all $k' \neq k$; and (iii) $\omega^{(k)}$ is independent of $y^{(k')}$ for  all $k'\in [K]$.

For each task, we assume that the set of non-zero coefficients is sparse and use penalized multi-task regression to identify the relevant features. We impose a specific inter-task structure on the joint regression by assuming that the coefficients
$\Theta^{(1)}, \dots, \Theta^{(K)} \in \mathbb{R}^p$
can be represented as the product of a common parameter that is shared between all tasks and a task-specific parameter that is unique to an individual task; namely, \begin{equation}
\label{decomposition:GL}
\Theta^{(k)}_{j} = \tau_j \ \gamma_j^{(k)}, \text{ for } j \in [p], k\in [K].
\end{equation}
A similar multiplicative parameterization has been used for joint estimation in a number of penalized MTL algorithms \citep{bi2008improved, lozano2012multilevel, wang2016multiplicative}.
To avoid a sign ambiguity, we take 
$\tau_j \geq 0  \text{ for } j \in [p].$
%The coefficients $\Theta^{(1)},\dots,\Theta^{(K)}$ can be uniquely identified in this parameterization.%  
 We do not impose any further constraints to ensure that the two components, $\tau_j$ and $\gamma_j^{(k)},$ are each identifiable since our interest lies in estimating the sparsity structure of their product.   Note that $\tau_j$ determines the sparsity at the global level as $\tau_j =0$ implies that $\Theta^{(k)}_{j}=0$ for all $k\in [K]$.   The parameter  $ \gamma_j^{(k)}$ controls the task-specific sparsity, with  $\gamma_j^{(k)} = 0$ indicating that $\Theta^{(k)}_{j}=0$ for task $k$.

We proceed to fit a sparse model by minimizing the penalized objective function 
%Relevant predictors, those for which the estimated coefficient in \eqref{decomposition:GL} is non-zero, are identified \color{black} through an $\ell_1$-regularized joint regression:
\begin{equation}
\label{canonical:MTL}
 %\h{\tau},\left\{ \h{\gamma}^{(k)}\right\}_{k\in [K]}  =  \argmin_{\{\tau_j\}_{j\in [p]}, \{\gamma_j^{(k)}\}_{j\in [p], k\in [K]}} 
  \sum_{k=1}^K \mathcal{L} \big( \Theta^{(k)}; y^{(k)}, X^{(k)}, \omega^{(k)}\big) + \eta_1 \sum_{j=1}^p \tau_j + \eta_2 \sum_{j=1}^p \sum_{k=1}^K |\gamma_j^{\left( k \right)}|,
\end{equation}
subject to the constraint $\tau_j\geq 0 \text{ for } j\in [p].$
The minimizers $\h{\tau}\in \mathbb{R}^p$, $\h{\gamma}^{(k)}\in \mathbb{R}^p$ will clearly both be sparse due to the $\ell_1$ penalties in the objective function.   %are the vectors obtained by concatenating $\widehat\tau_j$ and $\wide-hat{\gamma}_j^{(k)}$, respectively, for $j\in [p]$. The loss in the MTL objective \eqref{canonical:MTL} is a sum of the randomized loss functions corresponding to each task; the loss for the $k^{\text{th}}$ task is obtained by plugging the response vector $y^{(k)}$, predictor matrix $X^{(k)}$ and randomization variable $\omega^{(k)}$ into the randomized loss defined in \eqref{Randomized Loss}. The $\ell_1$ penalties encourage sparsity in both the common and task-specific coefficients at rates determined by the tuning parameters, $\eta_1$, $\eta_2\in \mathbb{R}^{+}$.
To optimize the objective \eqref{canonical:MTL}, we note that its solution will yield coefficients $\Theta^{(k)}$ for  $k\in [K]$ that could be obtained from an equivalent formulation 
\begin{equation}
\label{reformulated}
 \argmin_{\{\Theta_j^{(k)}\}_{j\in [p], k\in [K]}} \sum_{k=1}^K \mathcal{L} \big( \Theta^{(k)}; y^{(k)}, X^{(k)}, \omega^{(k)}\big) + 2\lambda \sum_{j=1}^p \left\{\sum_{k=1}^K |\Theta_j^{(k)}|\right\}^{1/2}, 
\end{equation}
where $\lambda= \color{black} \sqrt{\eta_1 \eta_2}$.     This equivalence was established in previous work; see, e.g.,  \cite{guo2011joint} and \cite{lozano2012multilevel}. 

Following the approach of \cite{zou2008onestep} and \cite{guo2011joint}, we use an iterative local linear approximation to the penalty term in \eqref{reformulated}, centered at the absolute value of the previous iterate:  
%We iteratively approximate the penalty on the composite regression coefficients in \eqref{reformulated} with a local linear analog centered at $\big(  \h{\Theta}_j^{\left( k \right) } \big)^{\left( t \right)}$, the solution of the previous $t^{\text{th}}$ iterate. That is,
$$ 2 \left\{\sum_{k=1}^K |\Theta_j^{(k)}|\right\}^{1/2} \sim c^{(t)} + \sum_{k=1}^K \dfrac{\color{black} |\Theta_j^{(k)}| \color{black}}{\sqrt{ \sum_{k=1}^K  \big| {\big(  \h{\Theta}_j^{\left( k \right) }} \big)^{\left( t \right)} \big|}}.$$
The constant $c^{(t)}$ depends only on the previous iterate and can be ignored in the corresponding optimization problem. With this approximation, the multi-task objective conveniently decouples by task. The successive estimates of the regression coefficients can be computed by solving a LASSO problem separately for each of the $K$ tasks, dependent on the previous iterate only through the penalty weights 
$$\lambda_j^{(t+1)} = \min \Big\{\lambda_0, \, \lambda \cdot \big( \sum_{k=1}^K  \big| {\big( \h{\Theta}_j^{\left( k \right)}} \big)^{\left( t \right)} \big|\big)^{-\frac{1}{2}} \Big\}, \;j\in [p] . $$
Here, $\lambda_0$ is a pre-specified large positive constant used for numerical stability, and $\lambda$ is a tuning parameter. The iterative procedure is summarized in Algorithm \ref{alg:Algselect}. 

\begin{algorithm}[H]
    \caption{Estimating Shared Sparsity}
  \label{alg:Algselect}
  \begin{algorithmic}[1]
    \vspace{0.05in}

             \For {k=1,\dots,K}
          \State Initialize $\big(\h{\Theta}^{( k)}\big)^{(0)} = \argmin_{\Theta^{(k)}} \left( \mathcal{L}\big(\Theta^{(k)}; y^{(k)}, X^{(k)}, \omega^{(k)}\big) + \lambda \displaystyle\sum_{j=1}^p | \Theta_j^{(k)}| \right)$
          \EndFor
    \Procedure {Iterate until convergence}{}
    \State Let $t=0$
    \While{convergence$<$tol}
      \For{$j = 1, \dots, p$}
         \State $\lambda_j^{(t+1)} = \min \Big\{\lambda_0, \, \lambda \cdot \big( \sum_{k=1}^K  \big| {\big( \h{\Theta}_j^{\left( k \right)}} \big)^{\left( t \right)} \big|\big)^{-\frac{1}{2}} \Big\}$
      \EndFor
       \For{$k = 1, \dots, K$}
        \State  \textbf{LASSO}:
        \State Solve $\big(\h{\Theta}^{( k)}\big)^{(t+1)} = \argmin_{\Theta^{(k)}} \left( \mathcal{L}\big(\Theta^{(k)}; y^{(k)}, X^{(k)}, \omega^{(k)}\big) +  \displaystyle\sum_{j=1}^p \lambda_j^{(t+1)} | \Theta_j^{(k)}| \right)$
      \EndFor
      \State $t=t+1$
      \EndWhile
    \EndProcedure
    \algstore{bkbreak}
     \end{algorithmic}
\end{algorithm}
\noindent  Upon convergence of Algorithm \ref{alg:Algselect}, we obtain 
\begin{equation}
\label{supp:MTL}
\h{E}_k =\text{Supp}(\h{\Theta}^{(k)}) = E_k, \text{ for } k\in [K].
\end{equation}
Let the cardinality of $E_k \subseteq [p]$ be equal to $\delta_k$, and let $\delta=\delta_1+ \cdots + \delta_K$.

\section{Maximum Likelihood Inference Post MTL}
\label{sec:MLE}

Next, we proceed to  specify a model with the sparsity structure estimated through \eqref{canonical:MTL} and  quantify uncertainty in the effects of selected predictors with respect to this model.    We restrict our search to linear models of the form 
\begin{equation}
\label{model:MTL}
%\V{y}^{(k)}\sim\text{N}\left(X^{(k)}_{E_k}\beta^{(k)}_{E_k}, \sigma_k^2\cdot \rI_{n_k}\right) \text{ for }  k\in [K],
y^{(k)}  =  X^{(k)}_{E_k}\beta^{(k)}_{E_k} + \varepsilon_k   \text{ for }  k\in [K],
\end{equation} 
where $\varepsilon_k$ is a vector of errors with length $n_k$, mean 0, and variance $\sigma^2_k$.  We assume that the errors are independent across samples and across tasks.  We fit this model by maximizing the corresponding normal likelihood, which can be viewed as a general $M$-estimation procedure.   For the inference step, the post-selection confidence intervals for $\beta^{(k)}_{E_k}$  will be based on the normal assumption for the error distribution.    

%We attempt to derive a selection-adjusted likelihood for $\beta^{(k)}_{E_k}, k \in [K]$ under the prevailing conditional prescription for $\ell_1$-regularized algorithms, but we are encumbered by the non-affine geometry of the multi-task selection region. The lack of clear means by which to proceed with inference under this existing prescription, discussed in detail below, necessitates the development of new post-selection inference procedures. We identify a refined characterization of the selection event through exactly $$q=q_1+ \cdots + q_K$$ very simple linear inequalities to circumvent the difficulties encountered under the prevailing prescription. In the same vein as \cite{panigrahi2019approximate}, we use the final characterization to develop feasible estimating equations for an approximate MLE and the Obs-FI matrix at this estimator which enjoy statistical consistency with respect to the exact likelihood. 

\subsection{Some preliminaries}
\label{sec:3.1}

%We now define some notational preliminaries that will arise in the details of our likelihood. Let $N = n_1 + \dots + n_k$. 

We use boldface capital letters to denote stacked quantities across tasks. Let $$\Stacked{Y} = \left( {y^{(1)}}'\dots{y^{(K)}}' \right)',  \ \ \Stacked{E} =\{E_1,\dots,E_k\},  \ \ \StackedSymbol{\beta}_{\Stacked{E}} = \left( {\beta_{E_1}^{(1)}}'\dots{\beta_{E_k}^{(K)}}' \right)'.$$
Without loss of generality, we can reorder each predictor matrix and randomization instance to have the active components precede the inactive components:
$$X^{(k)} = \begin{bmatrix} X_{E_k}^{(k)} & X_{-E_k}^{(k)} \end{bmatrix}, \ \ \omega^{(k)} = \begin{pmatrix} \omega_{E_k}^{(k)} \\ \omega_{-E_k}^{(k)}\end{pmatrix} , $$
where $-A$ denotes the complement of set $A$. For each predictor $j \in [p]$, we use $\widetilde{j}_k$ to denote the index of this predictor in task $k$ after the permutation. Let $b^{(k)} \in \mathbb{R}^{\delta_k}$ denote the absolute values of the estimated non-zero multi-task regression coefficients for task $k\in [K]$ under this permutation, i.e.,
 $$b_{\widetilde{j}_k}^{(k)} = |\h{\Theta}_{j}^{(k)}| \hspace{1mm} \text{ whenever } |\widehat\Theta_{j}^{(k)}| \neq 0,$$
 and let
 $$\Stacked{B} =\left( {b^{(1)}}'\dots {b^{(K)}}' \right)'.$$

We let $s^{(k)} \in \mathbb{R}^{\delta_k}$ and $u^{(k)}\in \mathbb{R}^{p-\delta_k}$ represent the active and inactive components of the subgradient vector of the $\ell_1$-norm for the multi-task regression coefficients when evaluated at the solution, i.e.,
$$\left({s^{(k)}}' \,\  {u^{(k)}}' \right)' = \mathcal{D}_{\widehat\Theta} \Big\| \left( {\Theta_{E_k}^{(k)}}'
\ \ {\Theta_{-E_k}^{(k)}}' \right)^{\prime}\Big\|_1   \text{ for } k\in [K].$$ 
Note that the vector $s^{(k)}$ gives the signs of the active multi-task coefficients for the $k^{\text{th}}$ task, and $u^{(k)}$ satisfies $\|u^{(k)}\|_{\infty} \leq  1$. To reference the active and inactive components, respectively, of all the evaluated $\ell_1$-norm sub-gradients, we define $$\Stacked{S} = \left( {s^{(1)}}' \dots{s^{(K)}}' \right)^{\prime}, \Stacked{U} = \left( {u^{(1)}}' \dots{u^{(K)}}' \right)'.$$

We next introduce some notation to account for information shared between tasks. Suppose there are $r$ predictors, $j_1,\dots,j_r,$ that are active in one or more tasks. Define $\StackedSymbol{\Gamma} \in \mathbb{R}^r$ by
$$\StackedSymbol{\Gamma} = \left( \Gamma^{(j_1)} \dots \Gamma^{(j_r)}\right)',$$
where
$$\Gamma^{(j)} =\sum_{k=1}^K |\h{\Theta}_{j}^{(k)} | \;\; \text { for } j \in \{j_1,\dots,j_r \}.$$ 
Let the set of tasks where predictor $j$ is active be given by $$\kappa(j) = \{k: |\widehat\Theta_{j}^{(k)}| \neq 0 \},$$ with $d_j = |\kappa(j)|$ and elements $\kappa_1<\dots<\kappa_{d_j}$ arranged in increasing order.  For the active predictors $j_1,\dots,j_r$, we define $v^{(j)}$ to be a vector representing the first $(d_j-1)$ corresponding coefficients: 
$$v^{(j)} = \begin{pmatrix}|\widehat\Theta_{j}^{(\kappa_{1})}| & \dots & |\widehat\Theta_{j}^{(\kappa_{{d_j-1}})}|\end{pmatrix}' \;\; \text { for } j \in \{j_1, \dots, j_r \}.$$
The vector $v^{(j)}$ collects the absolute value of the non-zero coefficients for predictor $j$ across tasks, excluding the coefficient for the last task where predictor $j$ is non-zero. For most predictors in the active set, we expect that $d_j\geq2$ due to the shared sparsity across tasks; however, if we estimate $d_j=1$, then $v^{(j)}$ is empty and can be disregarded. After accounting for all but the last active coefficient corresponding to each predictor, we let 
$$\Stacked{V} = \left( {v^{(j_1)}}' \dots {v^{(j_r)}}'\right)'. $$

Note that there is a bijective mapping between $\Stacked{B}$ and $\left(\Stacked{V}, \StackedSymbol{\Gamma} \right)$. We introduce a matrix
$D \in \mathbb{R}^{r \times (\delta-r)}$ to record which elements of $\Stacked{V}$ correspond to each of the $r$ active predictors, with rows given by 
$$D_{i} = \begin{pmatrix} \mathbf{0}_{(d_{j_1}-1)}' & \dots & \mathbf{1}_{(d_{j_i}-1)}' & \dots & \mathbf{0}_{(d_{j_r}-1)}' \end{pmatrix} \;\; \text{ for } i\in[r].$$ 
For some permutation matrix $\mathcal{A} \in \mathbb{R}^{\delta \times \delta}$, the relationship between $\Stacked{B}$ and $\left(\Stacked{V},\StackedSymbol{\Gamma}\right)$ is given by
$$\Stacked{B} = \mathcal{A} \begin{pmatrix} \Stacked{V} \\  \StackedSymbol{\Gamma} - D \Stacked{V}  \end{pmatrix}.$$ 
Let the matrix $H \in \mathbb{R}^{\delta \times (\delta-r)}$ and the vector $g \in \mathbb{R}^{\delta}$ be given by: $$H = \begin{pmatrix} \rm{I}_{(\delta-r)} \\
 -D \end{pmatrix}\;\;\; g = \begin{pmatrix} \mathbf{0}_{(\delta-r)} \\ -\StackedSymbol{\Gamma} \end{pmatrix}.$$
To enforce the $\delta$ linear inequalities given by $Hv \geq g$, we use the following barrier function:
$$\phi_{H,g}(v) = \begin{cases} \sum_{j=1}^{\delta} \log \left(1 + \frac{1}{H_j v- g_j}\right) & \text{ if } Hv > g \\
\infty & \text{else.}\end{cases},$$
where $H_j$ is the $j^{th}$ row of $H$ and $g_j$ is the $j^{th}$ component of $g$.

\subsection{Estimating Equations for Approximate MLE-Based Inference}
\label{sec:3.2}

A natural starting point for selective inference in the multi-task setting is the law for $\Stacked{Y},$ conditioning on the event 
\begin{equation}
\label{Lee:event}
\left\{\h{\Stacked{E}} = \Stacked{E}, \;\h{\Stacked{S}}=\Stacked{S}  \right\}.
\end{equation}
This conditional prescription results in practical selective inference procedures for other $\ell_1$-regularized algorithms by constraining the response to fall within a polyhedral partition of the sample space. Consider, for example, the randomized LASSO procedure described in section \ref{sec:two_stage}. 
The event \eqref{Lee:event}, characterized through the K.K.T. conditions, induces an affine map from the randomization variable $\omega$ to the absolute coefficients $b$ and subgradient $u$ at the optimal solution. Under this transformation, the Jacobian contributes only a proportionality constant to the conditional law of $\left(y,b,u \right)$ given \eqref{Lee:event}, yielding a multivariate Gaussian distribution truncated to a polyhedral partition. The distribution for $y$ that results from conditioning upon $u$ and marginalizing over $b$ can be used to facilitate approximate MLE-based inference \citep{panigrahi2019approximate}. 

Unfortunately, this conditional prescription does not generalize well to the MTL setting.  
Note that the stationary map for the model selection procedure in Algorithm \ref{alg:Algselect}, given $\Stacked{E}$ and $\Stacked{S}$, induces a relationship between $\Stacked{W} = \left( {\omega^{(1)}}'\dots{\omega^{(K)}}' \right)'$, the collection of randomization variables, and $\left(\Stacked{B}, \Stacked{U}\right)$. This relationship implies a transformation of the form:
  $$\Stacked{W} = \begin{pmatrix} \pi^{(1)}(b^{(1)},u^{(1)}) & \pi^{(2)}(b^{(2)},u^{(2)})  &\cdots & \pi^{(K)}(b^{(K)},u^{(K)})\end{pmatrix},$$
  where
\begin{equation}
\label{KKT:k}
\begin{aligned}
\pi^{(k)}(b^{(k)},u^{(k)}) = -{X^{(k)}}^{\prime} y^{(k)} & + \begin{bmatrix} {({X^{(k)}_{E_k}})}^{\prime} X^{(k)}_{E_k} + \epsilon \cdot \rm{I}_{\delta_k} \\ {({X^{(k)}_{-E_k}})}^{\prime} X^{(k)}_{E_k} \end{bmatrix} \Dg(s^{(k)})b^{(k)}  \\ &\;\;\;\;\;\;\;\;\;\;\;\;\;\;\;\;\;\;\;\;\;\;\;\;\;\;\;\;\;\;\;\;\;\;\;\;\;\;\; + \Dg(\Lambda^{(k)}) \begin{pmatrix}  s^{(k)} \\ u^{(k)} \end{pmatrix},
\end{aligned}
\end{equation}
$$\Lambda_{\widetilde{j}_k}^{(k)} = \min \Big\{\lambda_0, \, \lambda \cdot \big( \sum_{k=1}^K  \big| {\big( \h{\Theta}_j^{\left( k \right)}} \big|\big)^{-\frac{1}{2}} \Big\}, \;j\in [p].$$
Observe that this transformation is non-affine since the penalty term is now related to the solution. The change-of-variables Jacobian, given in Proposition 1 of the Appendix, is a complicated function of $(\Stacked{B},\Stacked{U})$.    Deriving estimating equations for maximum-likelihood inference using the law of $(\Stacked{Y},\Stacked{B}, \Stacked{U})$ conditional upon \eqref{Lee:event} would require closed-form expressions for the partial derivatives of the Jacobian, which are not available.
For completeness sake, we provide the likelihood based on this conditional law in Proposition 2 under the Appendix.
%Hence, this common conditional prescription will not yield tractable MLE-based inference procedures for our multi-task model selection procedure. 

Instead, we propose a different approach that can  bypass the intractable Jacobian and avoid cumbersome numerical integrations to easily facilitate maximum likelihood inference.   We will work with an exact selection-adjusted likelihood, derived by conditioning on the refined event
\begin{equation}
\label{new:conditioning:event}
\left\{\h{\Stacked{E}} = \Stacked{E},\ \h{\Stacked{S}}= \Stacked{S},\ \h{\StackedSymbol{\Gamma}}=\StackedSymbol{\Gamma}, \  \h{\Stacked{U}}=\Stacked{U}\right\},
\end{equation}
which is a proper subset of the event in \eqref{Lee:event}.
We form our estimating equations for maximum likelihood inference in terms of the least squares estimator based on the selected predictors for each task, 
\begin{equation}
    \label{mle:naive}
\h{\beta}_{E_k}^{(k)}= {(X^{(k)}_{E_k})}^{\dagger}y^{(k)}.
\end{equation}
Note that the estimator in \eqref{mle:naive} is the naive MLE that we would have used if the sets $E_k$ were specified before looking at the data.
Dependent on ${X^{(k)}}' y^{(k)}$, the event of selection also relies on  
$$\h{\beta}_{\perp}^{(k)} =  {(X^{(k)})}^{\prime}\left({\rm{I}_{n_k}}- X^{(k)}_{E_k}{(X^{(k)}_{E_k})}^{\dagger}\right) y^{(k)},$$
the ancillary statistic we obtain through a projection of the response onto the subspace orthogonal to the span of the selected predictors $X^{(k)}_{E_k}$. 

Lemma \ref{ref:conditioning:event} first identifies an equivalent representation for the refined conditioning event in terms of $\Stacked{V}$, $\StackedSymbol{\Gamma}$ and $\Stacked{U}$ that we observe after solving the MTL Algorithm \ref{alg:Algselect};  please see Section \ref{sec:3.1} for a complete list of definitions. 

\begin{lemma} 
\label{ref:conditioning:event}
Suppose $\Stacked{V}$, $\StackedSymbol{\Gamma}$ are defined as above. Then the event \eqref{new:conditioning:event} is equivalent to the event
$$ \left\{\Stacked{V}>\mathbf{0}, \  \h{\StackedSymbol{\Gamma}}=\StackedSymbol{\Gamma}, \  \StackedSymbol{\Gamma} - D\Stacked{V} > \mathbf{0} , \ \h{\Stacked{U}}=\Stacked{U}\right\} .$$
\end{lemma}
The proof can be found in the Appendix.  Consider the bijective mapping $\Psi_{\mathcal{A}}: \mathbb{R}^{K \cdot p}\to \mathbb{R}^{K \cdot p}$ such that
$$ 
\begin{pmatrix} \Stacked{B}' & \Stacked{U}' \end{pmatrix}'=  \Psi_{\mathcal{A}}\begin{pmatrix} \Stacked{V} & \StackedSymbol{\Gamma} & \Stacked{U} \end{pmatrix} = \Dg \left(\mathcal{A}, \rm{I}_{(K p \ - \ \delta)}\right)\begin{pmatrix} \Stacked{V}' & \left(\StackedSymbol{\Gamma} - D \Stacked{V}\right)^{\prime} & \Stacked{U}^{\prime} \end{pmatrix}'.
$$
Applying a change of variables via the composite mapping,
\begin{equation}
\label{CoV}
\Stacked{W} \stackrel{\left(\Pi_{\Stacked{X}' \Stacked{Y}}\circ\Psi_A\right)^{-1}}{\xrightarrow{\hspace*{1.25cm}}}  \begin{pmatrix}\Stacked{V} &\StackedSymbol{\Gamma} & \Stacked{U} \end{pmatrix},
\end{equation}
we obtain an exact selection-adjusted likelihood in Theorem \ref{final:conditional:law} after conditioning on the event in \eqref{new:conditioning:event}.
  The alternate characterization for the refined conditioning event in terms of the new variables $\Stacked{V}$, $\StackedSymbol{\Gamma}$ and $\Stacked{U}$ yields a selection-adjusted likelihood function that no longer involves the Jacobian, the term which previously hindered our attempts to solve the estimating equations. The normalizing constant for our refined conditioning event is simply a Gaussian integral over a support set that is characterized by exactly $\delta$ linear inequalities.

\begin{theorem}
\label{final:conditional:law}
Consider the model in \eqref{model:MTL}. The likelihood obtained from the law of the least squares estimates based on 
$\left\{X^{(k)}_{E_k},y^{(k)}\right\}_{k=1}^{K}$ after conditioning upon the event in Lemma \ref{ref:conditioning:event} is given by
\begin{equation*}
\begin{aligned}
& \left(\int \rho\left(\localvar{\beta}; L\StackedSymbol{\beta}_{\Stacked{E}} +m, \Sigma \; \right) \cdot \rho( \localvar{V} ; P\localvar{\beta} + q, \Delta \;) \cdot 1(H\localvar{V}\geq g) \; d\localvar{V} d\localvar{\beta} \right)^{-1} \cdot \rho\left(\h{\StackedSymbol{\beta}}_{\Stacked{E}} ; L\StackedSymbol{\beta}_{\Stacked{E}} + m, \Sigma \; \right).
\end{aligned}
\end{equation*}
\end{theorem}
%5We have hence derived an exact theoretical likelihood that accounts for the selected model by mathematically characterizing and conditioning upon the selection event. 
%Consistent with prior work in selective inference, our conditioning set contains additional information beyond the active sets $E^{(1)},\dots,E^{(K)}$ as needed to simplify the geometry of the selection region.
Expressions for the matrices $L$, $m$, $\Sigma$, $P$, $q$, and $\Delta$ are provided in the Appendix. 
To develop an easily solvable system of estimating equations for the MLE and the inverse observed Fisher information matrix, $\InvFI$, %\liza{where was this defined?  Also, awful notation, just write something like "observed Fisher Information $\hat I$.}, 
we bypass the integration in the normalizer, simply approximating it with the mode of the integrand in the selection region.
That is, 
\begin{equation}
\begin{aligned}
& \log \int \rho\left(\localvar{\beta}; L\StackedSymbol{\beta}_{\Stacked{E}} + m, \Sigma \; \right) \cdot \rho(\localvar{V}; P\localvar{\beta} + q, \Delta) \cdot 1(H\localvar{V}\geq g) \; d\localvar{V} d\localvar{\beta}\\
&\approx -\inf_{\localvar{\beta}, \localvar{V} } \; \Big\{\dfrac{1}{2}(\; \localvar{\beta} - L\StackedSymbol{\beta}_{\Stacked{E}} -m)'\Sigma^{-1}(\; \localvar{\beta} - L\StackedSymbol{\beta}_{\Stacked{E}} -m) \\
&\;\;\;\;\;\;\;\;\;\;\;\;\;\;+ \dfrac{1}{2}(\localvar{V}- P \localvar{\beta} - q)'\Delta^{-1}(\localvar{V} - P \localvar{\beta} - q)+  \phi_{H,g}(\localvar{V}) \Big\}
\label{approx}
\end{aligned}
\end{equation}
ignoring an additive constant.
The approximation in \eqref{approx} then lends itself towards tractable equations for the selective MLE, $\Mb_{\Stacked{E}}$, and the inverse observed Fisher information matrix, $\InvFI$, given in Theorem \ref{est:equations:mle}. 
 
\begin{theorem}
\label{est:equations:mle}
Under the modeling assumptions in Theorem \ref{final:conditional:law}, the approximate selective MLE and the observed information matrix satisfy the following system of estimating equations.
$$\Mb_{\Stacked{E}} = L^{-1}\h{\StackedSymbol{\beta}}_{\Stacked{E}} + L^{-1}\Sigma P' \Delta^{-1}(P\h{\StackedSymbol{\beta}}_{\Stacked{E}} + q - \h{\Stacked{V}}) -L^{-1}m,$$
$$\InvFI=L^{-1} \Sigma L'^{ -1}  + L^{-1} \Sigma\left(P' \Delta^{-1}P - P' \Delta^{-1} \left(\Delta^{-1} + \grad^2 \phi_{H,g}(\h{\Stacked{V}}) \right)^{-1}\Delta^{-1}P \right)  \Sigma L'^{-1},$$
where $\h{\Stacked{V}}$ is obtained from solving 
\begin{equation}
\label{eq:opt}
\h{\Stacked{V}} = \underset{\localvar{V}}{\text{argmin}} \; \frac{1}{2}(\localvar{V}-P\h{\StackedSymbol{\beta}}_{\Stacked{E}} - q)^{\prime}\Delta^{-1} (\localvar{V}-P\h{\StackedSymbol{\beta}}_{\Stacked{E}} - q) + \phi_{H,g}(\localvar{V}).
\end{equation}
\end{theorem}
With these estimators for the approximate MLE and the inverse observed Fisher information matrix, it is now possible to use maximum likelihood inference to form $100\cdot(1-\alpha)\%$ confidence intervals for the parameters within the MTL model \eqref{model:MTL}. Algorithm \ref{alg:Alginference} summarizes  our procedure for post-MTL inference by reusing the same data.  %\liza{Not sure this needs these weird line numbers.   It also doesn't match the format of Algorithm 1, which is all in pseudo-code.   I would rather get rid of line numbers in both, though numbering steps would be ok, and convert the first two steps of Alg 2 into input to Alg 2, which is not part of the algorithm.    } 

\begin{algorithm}[H]
  \caption{Multi-Task Model Selection and Inference}
  \label{alg:Alginference}
  \begin{algorithmic}[1]
  \algrestore{bkbreak}
  \State Record the values for $\h{\Stacked{E}},\h{\Stacked{S}}, \h{\StackedSymbol{\Gamma}}, \text{ and } \h{\Stacked{U}}$ at convergence of Algorithm \ref{alg:Algselect}
  \State From the estimated sparsity structure, compute $\h{\beta}_{E_k}^{(k)} \text{ and } \h{\beta}^{(k)}_{\perp}$ for $k\in [K]$ 
  \State Specify a significance level $\alpha$
  \State Compute $P$, $q$, and $\Delta$
  \State Optimize \eqref{eq:opt} with gradient descent to compute $\h{\Stacked{V}}$
  \State Compute the matrices $L$, $m$, $\Sigma$
\Procedure {Maximum Likelihood Inference}{}
\State Find $\Mb_{\Stacked{E}}$ and  $\InvFI$ based on the estimating equations in Theorem \ref{est:equations:mle}
\For {$j \in \{1, \dots, \delta\}$}
\State Compute interval $\Mb_{{\Stacked{E}},j} \pm z_{1-\alpha/2}  {\sqrt{ \h{\boldsymbol{\mathsf{I}}}_{jj}^{\; -1}(\h{\StackedSymbol{\beta}}_{\Stacked{E}}^{{{\text{\; MLE}}}}) }}$
\EndFor
\EndProcedure
     \end{algorithmic}
\end{algorithm}

\section{Analysis of Neurocognitive Data from the ABCD Study}
\label{sec:ABCD}
\smallskip
The Adolescent Brain Cognitive Development (ABCD) study, discussed in the introduction, is a large longitudinal study undertaken to characterize typical cognitive development in adolescence \citep{luciana2018adolescent, volkow2018conception, karcher2021abcd}. Researchers are interested in better understanding the organization of cognitive abilities during childhood development, i.e., whether there is a single, general ability or separable fluid and crystallized abilities. Existing work has relied almost exclusively on behavioral data and has reached equivocal findings. Some studies have found evidence for a single strong, general factor of cognitive ability in youth \citep{juan2000testing, gignac2014dynamic}. Other studies have concluded that a two-factor model of intelligence performs better than a single-factor model for children and adolescents, with fluid intelligence and crystallized intelligence becoming more differentiated with age \citep{hulur2011intelligence, simpson2020neurocognitive}.  

Studies that leverage neurological data in studying the organization of cognitive abilities are extremely rare. Two exceptions are \cite{tadayon2020differential} and \cite{simpson2020neurocognitive}, who examined cortical morphology and white matter, respectively. These structural modalities, however, tend to have much weaker associations with cognitive abilities than resting-state fMRI \citep{marek2022reproducible, chen2022shared}. In addition, these studies used traditional mass univariate approaches, treating all edge weights as a bag of features.  We instead adopt a multi-task approach that, as we have argued, is better suited for the dual goals of prediction and cartographic mapping.  

\subsection{Multi-Task Framework for Studying ABCD Data}
\label{sec:5.1}

Applied to the ABCD data, our multi-task learning and selective inference procedure, MTL + SI, offers a novel approach to identifying the neurobiological underpinnings of fluid and crystallized intelligence in the developing brain.  %. Multi-task learning can improve detection of any shared neurological signals across tasks, while selective inference provides additional assurance that the signals recovered through multi-task learning did not arise spuriously
We have applied MTL + SI to the second ABCD release (the study is ongoing), using the same inclusion/exclusion criteria as \cite{sripada2021brain-wide}. This leaves data for 5,937 subjects from 19 different research sites throughout the US. We limit the analysis to four of the tasks from the ABCD neurocognitive battery, two measuring fluid intelligence (the List Sorting Working Memory Test  from the NIH Toolbox for the Assessment of Neurological and Behavioral Function and the Matrix Reasoning subtest from the Wechsler Intelligence Test for Children)  and two measuring crystallized intelligence (the Picture Vocabulary Test and the Reading Comprehension Test from the same NIH Toolbox).  

The predictors in our multi-task regression are neurological factors extracted from resting-state fMRI data by estimating the connections between 418 regions of interest (ROIs, or nodes) in the brain. These ROIs were identified based on the Gordon cortical parcellation \citep{gordon2016generation}, augmented with additional subcortical and cerebellar atlases. The 418 ROIs are further classified into 15 functional groups, ranging in size from 4 to 54 nodes; see Table \ref{tab:Networks} for a complete list. These groups of ROIs, equivalent to communities in the statistical network analysis literature, are themselves called networks in neuroimaging, an unfortunate terminology overload.

To assess the strength of connection between each pair of ROIs, the Pearson correlation coefficient is computed between the fMRI blood oxygen level dependent (BOLD) signals at those ROIs. We use the  inter-node correlations from \cite{sripada2021brain-wide}, computed after pre-processing the fMRI time series data to correct for nuisance covariates, such as physiological noise and head motion, through a standard pipeline that includes FreeSurfer normalization, ICA-AROMA denoising, CompCor correction, and omission of high-motion frames. With 418 nodes, each scan corresponds to 418 $\times$ 417 / 2 = 87,153 features. A standard approach to resting-state fMRI data is to replace the edge weights with top principal component scores \citep{cordes2006estimation}, which both reduces dimensionality and helps with the low signal-to-noise ratio.  We retain the first 500 principal component scores computed from the correlations to use as features in multi-task learning.   Previous work has shown that a small number of components is typically sufficient to capture most inter-individual variation in functional connectivity \citep{sripada2019basic} and predict differences in cognition \citep{sripada2020prediction}.

After standardizing the response variables, we consider two different ways of identifying the principal components for each task through multi-task learning. One approach, which we refer to as joint, is to apply Algorithm \ref{alg:Algselect} to all four tasks. This is the default approach to multi-task learning on both the simulated and real data when not otherwise specified. The alternative approach, which we refer to as pairwise, is to apply Algorithm \ref{alg:Algselect} separately to the two tasks that measure fluid intelligence and the two tasks that measure crystallized intelligence.  Comparing the two approaches allows us to assess how much joint learning of the fluid and crystallized tasks can improve detection of shared neurological structure. We proceed to fit a multi-task model on the selected principal components for each approach and construct confidence intervals for the best linear coefficients using Algorithm \ref{alg:Alginference}. A consistent plug-in estimator is used to approximate the noise level, following the recommendation of \cite{randomized_response}.  We use 80\% of the original data for model selection and inference, hold out another 10\% for selecting the tuning parameter for penalized multi-task regression, and use the last 10\% as test data.

\subsection{Validation of MTL+SI Through Simulation}
\label{sec:5.2}

We first use synthetic data with the same dimensions and estimated sparsity level as the fMRI features to investigate the efficacy of MTL + SI in recovering signals and estimating their strength. We generate data from the linear regression model in \eqref{model:MTL} with $K=4$ and noise variance $1$.
Each predictor matrix is simulated by drawing $6000$ samples from a multivariate Gaussian distribution of dimension $p=500$, with all means equal to zero and covariance given by the identity matrix.
Simulation results with non-Gaussian errors are reported in Section 3 of the Appendix. 

The coefficients $\StackedSymbol{\beta}\in\mathbb{R}^{p\times K}$ in the MTL model are chosen based on two parameters, the global sparsity level $s_G$ and the task-specific sparsity level $s_T$, both numbers between 0 and 1. We define the global sparsity as the percentage of predictors that are zero for all of the tasks. The task sparsity, meanwhile, quantifies the average number of tasks that do not share any one of the global signals. Note that a task-level sparsity value $s_T = 0$ would mean that all tasks have the same active predictors, and higher levels of $s_T$ indicate that the predictors used by each task are more heterogeneous. For given levels of $s_G$ and $s_T$, the entries of $\boldsymbol{\beta}$ are determined as follows. 

\begin{enumerate}
\item Select $[s_G p]$ predictors to be globally null for every task, and set the corresponding rows of $\boldsymbol{\beta}$ to zero.  
\item  For the remaining set of globally active predictors,  set an average of $[s_T K]$ entries of each row of $\boldsymbol{\beta}$ to zero.  
\item For each predictor $j\in[p],$ specify the corresponding non-zero coefficients by sampling without replacement from an equally-spaced sequence of values covering the interval $$\big[\sqrt{2\log(p)},\sqrt{6\log(p)}\big].$$ Randomly assign each coefficient a sign drawn from $\{-1,1\}$. 
\end{enumerate}

The task sparsity level is of primary scientific interest for its ability to measure the extent of feature-sharing between tasks, so we investigate the consequences of varying $s_T$ while fixing $s_G$. Our results from applying MTL+SI jointly to the four ABCD tasks, presented in the following sections, indicate that 77 of the top 500 principal components are significant for one or more tasks following selective inference, so we fix $s_G=0.85$. We estimate that the task sparsity level of the real data is $s_T=0.416$, but we consider three different task sparsity levels for our simulations: 0.25, 0.375, and 0.5. The simulations with a task sparsity level of 0.25 and 0.5 are designed so that each predictor in the globally active set is shared by three tasks and two tasks, respectively. For the simulations with a task sparsity level of 0.375, the globally active predictors are specified so that half are shared by three tasks and the remaining half are shared by two tasks. We conduct 100 simulation repetitions under each of the three data-generating models, corresponding to the three different sparsity levels. The tuning parameter is chosen to minimize the average MSE on the validation set across iterations, and the results are reported at that tuning parameter value for all 100 iterates. 

To assess both multi-task learning and selective inference when there is shared structure across tasks, we compare MTL + SI against two alternative approaches, MTL with data splitting and single-task LASSO with selective inference. Each procedure is described below.
\begin{enumerate}
\item \textbf{MTL($v$) + SI}:  Our proposed approach. Use Algorithm \ref{alg:Algselect} to select an MTL model based on independent Gaussian randomization variables, $\omega^{(k)} \sim N(0, v^2\sigma^2\cdot \rm{I})$ for $k\in [K]$, and construct selective inference (SI) confidence intervals by Algorithm \ref{alg:Alginference}.
\item \textbf{Data Splitting (DS($s$))}: Divide the training data into two parts, using $[s n_k]$ samples from each task $k \in [K]$ for model selection with the usual multi-task algorithm, and reserve the rest for inference about the selected predictors as described in \cite{cox}.
\item \textbf{LASSO($v$) + SI}:  Apply the randomized LASSO separately to each task using independent randomization variables, $\omega^{(k)} \sim N\left(0, v^2\sigma^2 \cdot \rm{I}\right) \text{ for } k\in [K]$, and proceed with SI using the maximum likelihood approach of \cite{panigrahi2019approximate}. 
\end{enumerate}

We consider two common choices for the data splitting parameter, $s=0.5$ and $s=0.67$. \cite{selective_bayesian} show that there is a rough equivalence between data splitting and a randomization variable with variation parameter $ v = \sqrt{(1-s)/s}$ for a Gaussian response. We find that the variation parameter $v=0.7$, corresponding to $s=0.67$, strikes a slightly more optimal balance between the quality of model selection and the quality of inference than the variation parameter $v=1.0$, corresponding to $s=0.5$. Thus, we set $v=0.7$ for all selective inference methods.

We use several metrics to evaluate the performance of each method. The empirical coverage rate (CR) of the confidence intervals is computed for the non-zero coefficients, defined by 
$$
\text{CR}= 1 - \frac{\lvert \{ j \in \h{\Stacked{E}}: \StackedSymbol{\beta}_{{\Stacked{E}}, j} \not\in \bm{C}_{\h{{\Stacked{E}}},j}\}\rvert} {\max(|\h{\Stacked{E}}|,1)}. $$
This rate is further averaged over replications. We assess the inferential power of each method by reporting the confidence interval lengths for the selected parameters. To measure the overall accuracy of model selection and subsequent inference, we compute the F1 score,  defined as 
$$
\text{F1} = 2  \dfrac{\text{precision} \times \text{recall} }{\text{precision} + \text{recall}}.
 $$
For our purposes, precision is the proportion of truly active predictors among those that were both selected into the model and deemed significant, that is, the effects with confidence intervals that did not cover zero. Recall is the proportion of all truly active predictors that were both selected into the model and deemed significant after inference. Letting $\Stacked{E}_0$ be the set of true active predictors, 
$$\text{Precision }=\frac{\lvert \Stacked{E}_0 \; \cap \; \{ j \in \h{\Stacked{E}}: 0 \not\in \bm{C}_{\h{{\Stacked{E}}},j} \} \rvert}{\lvert \{ j \in \h{\Stacked{E}}: 0 \not\in \bm{C}_{\h{{\Stacked{E}}},j}\} \rvert} \, ;  \;
\text{Recall }=\frac{\lvert \Stacked{E}_0 \; \cap \; \{ j \in \h{\Stacked{E}}: 0 \not\in \bm{C}_{\h{{\Stacked{E}}},j}\} \rvert}{\lvert \Stacked{E}_0 \rvert}. $$

Figure \ref{fig:vary_task_sparsity} shows the distribution of coverage, interval length, and F1 score for MTL(0.7) + SI and the alternative methods. We observe that all methods achieve a nominal coverage level of 90\%. In terms of interval length, MTL(0.7) + SI has a large advantage over DS(0.67), which asymptotically reserves a similar amount of information for inference. MTL(0.7) + SI also has similar or better performance than DS(0.5) in terms of interval length, even though the latter method asymptotically reserves more information for inference. Of the metrics we report, the F1 score provides the most direct comparison between the approaches, capturing both the validity of the chosen model and the significance of the results. We find that MTL(0.7) + SI achieves a higher median F1 score than all three alternatives for each task sparsity level. This indicates that similar tasks should be trained together whenever possible and that selective inference is a more optimal method for quantifying the uncertainty in models chosen through joint learning than sample splitting. Additional simulations which vary other parameters in the design---dimensions and global sparsity---are collected in the Appendix; please see Section 3. 

\begin{figure}[h!]
  \centering
    \includegraphics[width=\linewidth]{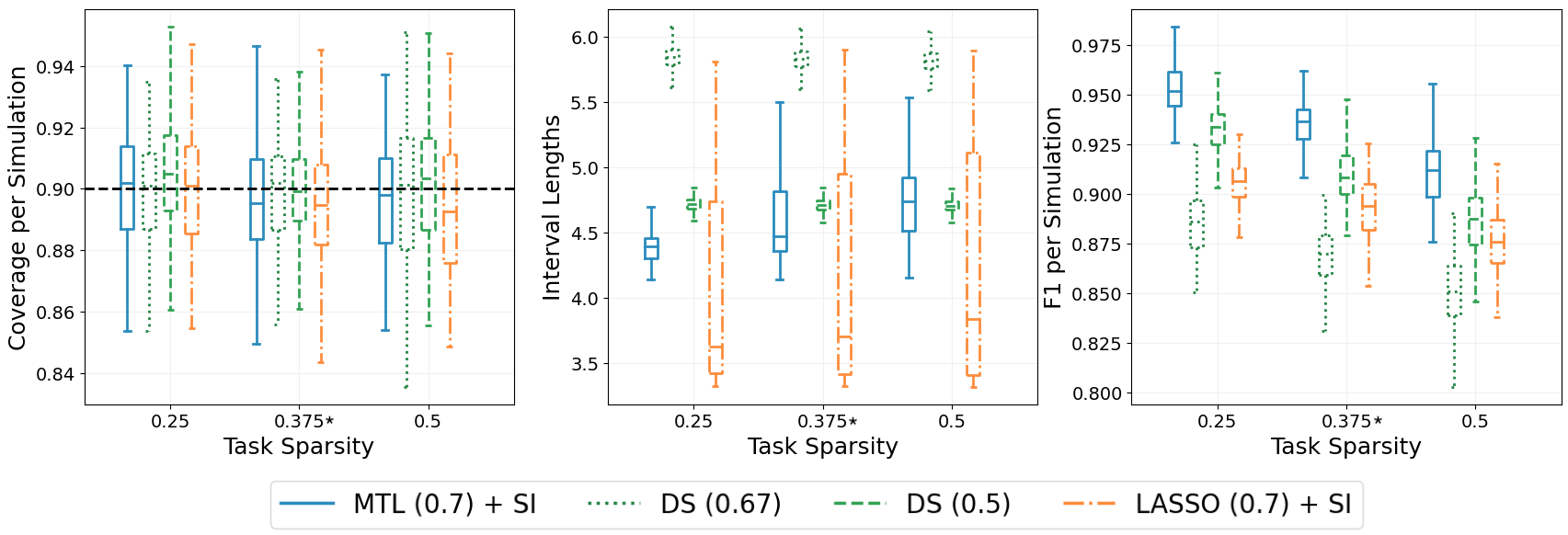}
  \caption{Comparison of MTL(0.7) + SI against three alternative approaches: DS(0.67), DS(0.5), and LASSO(0.7) + SI, with 100 replications for each level of task sparsity $s_T$. Global sparsity is fixed at $s_G = 0.85$. The asterisk (*) denotes the task sparsity level closest to the fMRI data. Outlier points are not shown to improve readability. Left: coverage; Middle: interval length; Right: F1 score.}
  \label{fig:vary_task_sparsity}
\end{figure}

Finally, we compare the results of applying MTL(0.7) + SI jointly to four simulated tasks against the results of applying MTL(0.7) + SI separately to two pairs of simulated tasks when a) there is shared structure across all tasks and b) when the shared structure is only present within the two pairs. To test setup (a), we randomly assign each active predictor to two tasks, ensuring that any two of the four tasks have roughly the same number of common predictors. To test setup (b), we instead specify the active predictors so that the two related tasks share a common set of predictors and the unrelated pairs have no common predictors. Note that $s_T=0.5$ in both setups, with only the relationship between tasks changing. Figure \ref{fig:joint_vs_separate_sim} indicates that the joint approach produces shorter intervals and a higher median F1 score under setup (a), while the joint and pairwise approaches have a similar median confidence interval length and F1 score under setup (b). These results confirm that the joint approach improves the quality of model selection and inference when there is common structure across tasks, without doing any real harm when there is no shared structure.

\begin{figure}[h!]
  \centering
    \includegraphics[width=\linewidth]{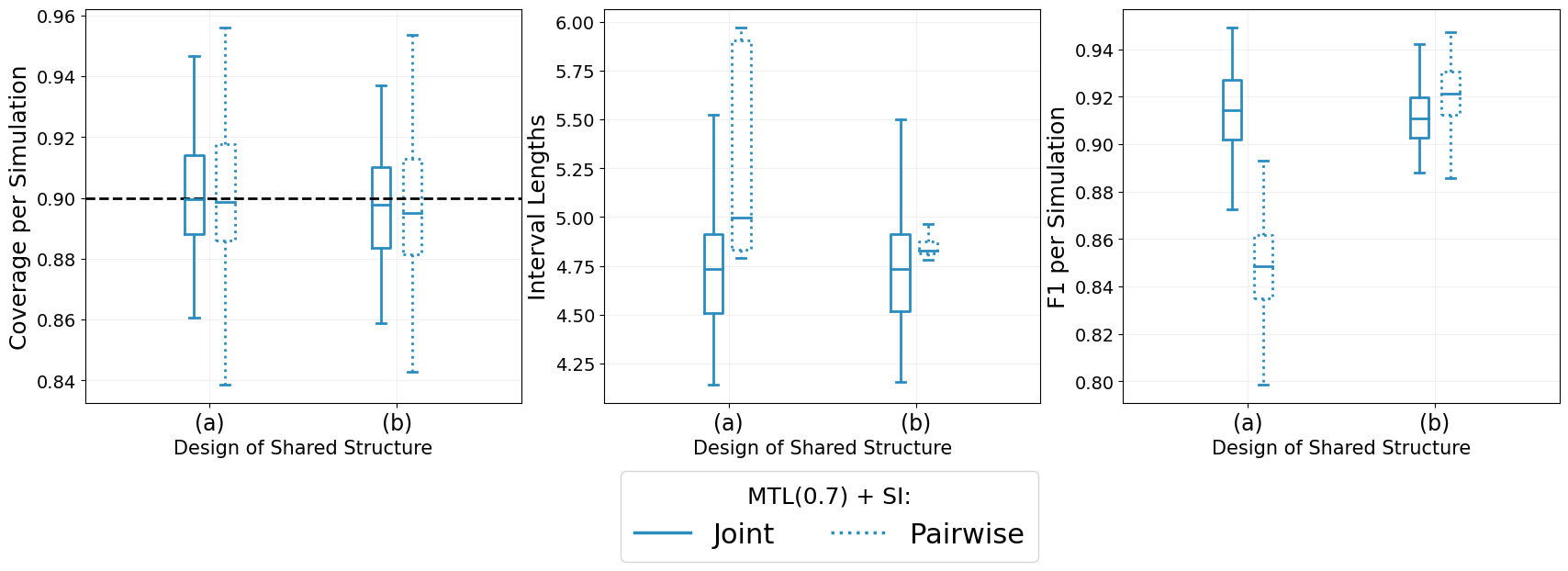}
  \caption{Comparison of the joint and pairwise approaches under setup (a), where any two tasks have roughly the same amount of common structure, and setup (b), where only the two pairs have common structure. Global and task sparsity are fixed at $s_G = 0.85$ and $s_T=0.5$. Outliers are not shown. Left: coverage; Middle: interval length; Right: F1 score.}
  \label{fig:joint_vs_separate_sim}
\end{figure}

\subsection{Assessing the Neurological Organization of Different Cognitive Abilities}
\label{sec:5.3}

We now apply our methods to the ABCD data, again comparing the joint and pairwise approaches to MTL(0.7) + SI. The left panel of Figure \ref{fig:Joint vs Separate} shows the predictive performance of each method on validation data. Following the convention in the neuroimaging literature \citep{sripada2020prediction, sripada2021brain-wide}, we measure the predictive performance of each method by the so-called predictive $r$, the correlation between predictions and observed responses on out-of-sample data. All tuning parameters are chosen to maximize the average predictive $r$ across tasks on the validation data. Table \ref{Tab:2} reports the final performance of each method on the testing data. The joint approach generally maintains some advantage across all tasks, and the observed correlations are consistent with the expectations of domain experts and previous findings \citep{sripada2021brain-wide}. 

The benefit of training all four tasks together is even more apparent when comparing the inferential power of the two methods. The right panel of Figure \ref{fig:Joint vs Separate} shows the confidence interval lengths obtained under each approach. We observe that applying MTL(0.7) + SI jointly to all four tasks results in a substantially smaller median confidence interval length than applying MTL(0.7) + SI separately to the two crystallized and the two fluid intelligence tasks. Overall, the results indicate that the joint approach can better detect some neurological features while also matching or exceeding the predictive performance of the two separate multi-task models for fluid and crystallized intelligence. 

\begin{table}[ht]
\centering
\begin{tabular}{lcc}
  \hline \hline
 & Joint Predictive $r$  & Pairwise Predictive $r$ \\ 
  \hline
{\text{NIH Tlbx RC}} & 0.335 & 0.321 \\
{\text{NIH Tlbx PV}} & 0.470 & 0.451 \\
{\text{Matrix Reasoning}} & 0.283 & 0.273 \\
{\text{NIH Tlbx LS}} &0.354 & 0.317 \\
   \hline
\end{tabular}
\caption{Predictive correlations computed on the testing data for the joint and pairwise approaches, where MTL(0.7) + SI is applied, respectively, to all four tasks or separately to the pairs of fluid/crystallized tasks.}
\label{Tab:2}
\end{table}

\begin{figure}[h]
    \centering
    \includegraphics[width=\linewidth]{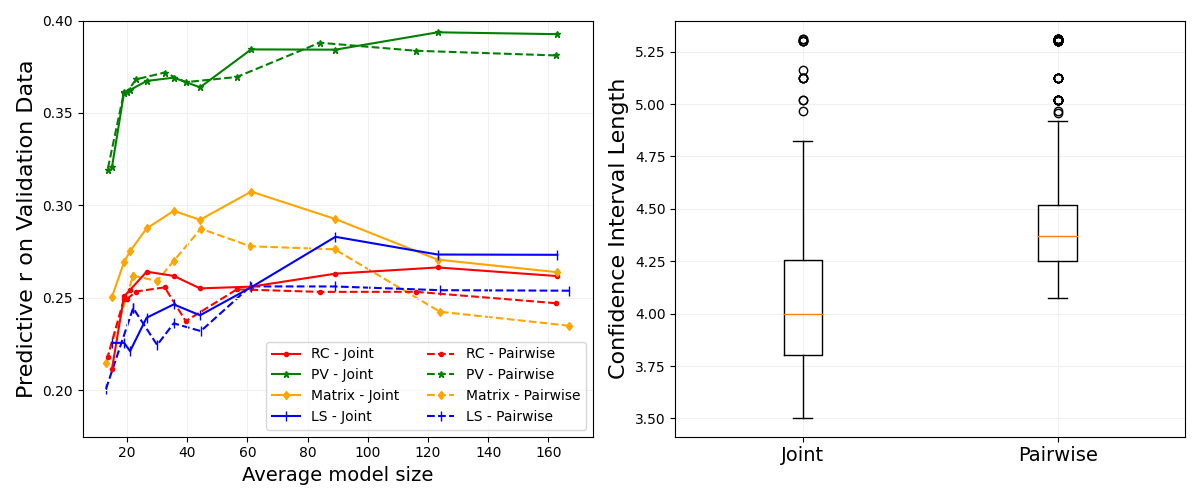}
  \caption{Left: The predictive $r$ for the joint and pairwise approaches to MTL(0.7) + SI. Right: The distribution of confidence interval lengths across tasks for the joint and pairwise approaches to MTL(0.7) + SI.}
  \label{fig:Joint vs Separate}
\end{figure}

Although the predictive advantage of the joint approach is slight, the chosen features are quite different from those recovered through the pairwise approach. Training all of the tasks together yields four very similar models, indicating that many of the same neurological features could underlie different cognitive abilities. To quantify the structural overlap between any two tasks, we use the Jaccard index to measure the similarity of the significant features for those tasks, where significant features are the PCs with confidence intervals that do not contain zero. Figure \ref{fig:Jaccard} shows the Jaccard similarity of the four tasks under both the joint and the pairwise approaches to model selection and inference. When applied separately to the two pairs of tasks, MTL(0.7) + SI reveals a significant amount of shared structure between the two crystallized intelligence tasks and the two fluid intelligence tasks, as well as a lesser but still notable amount of shared structure across the two pairs of tasks. The joint approach to model selection and inference, however, reveals a significant amount of common structure between any two of the four tasks. Selective inference seems to confirm the existence of the shared signals identified through joint model selection, suggesting that there may be a common neurological basis for different cognitive abilities that can be best understood through joint multi-task learning.

\begin{figure}[h]
    \centering
    \includegraphics[width=\linewidth]{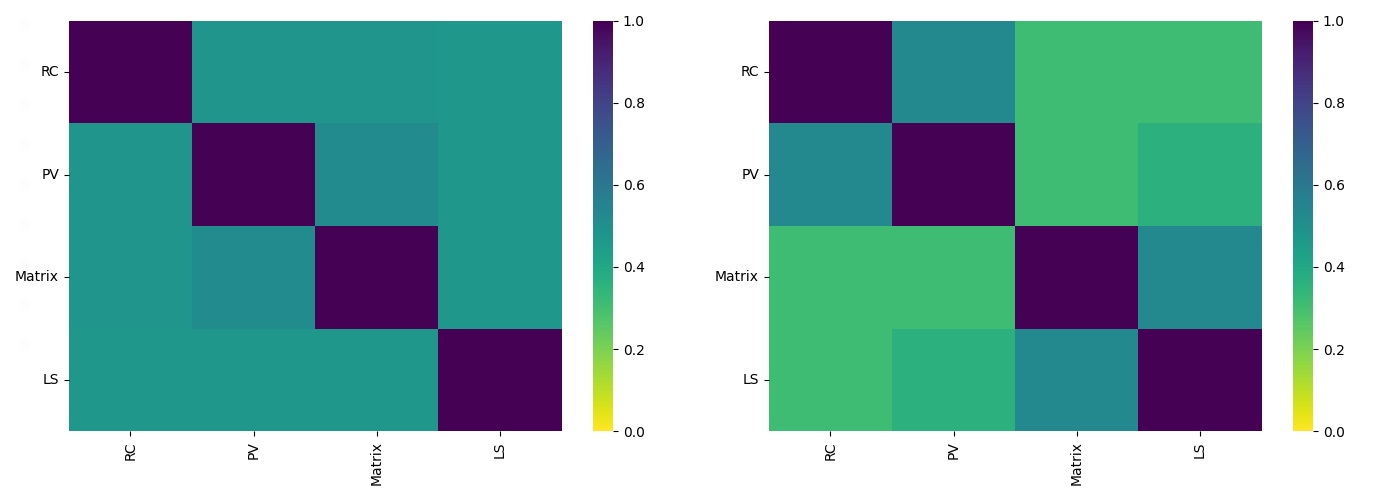}
  \caption{Similarity of selected models for the four tasks. Left: Jaccard index between sets of significant PCs recovered through the joint approach.  Right: Jaccard index between sets of significant PCs recovered through the pairwise approach.}
  \label{fig:Jaccard}
\end{figure}

We have also compared MTL(0.7) + SI to data splitting, applying DS(0.67) and DS(0.5) jointly to all four tasks. MTL(0.7) + SI seems to improve model selection relative to data splitting, with more of the selected features surviving inference. As shown in Figure  \ref{fig:inference_results}, MTL(0.7) + SI also yields a shorter median confidence interval length than data splitting, both overall and for predictors that are selected by both MTL(0.7) + SI and either DS(0.67) or DS(0.5). Note that some of the predictors vary between approaches since each method performs its own model selection; however, rough comparisons may still be possible when there is substantial overlap in the selected sets. Table \ref{tab:avg_results} reports the average number of features selected and further deemed significant across the four tasks for each method, as well as the average number of shared features between selective inference and each data splitting procedure. The overlap is significant, making the comparison more meaningful.

\begin{table}[ht]
\centering
\begin{tabular}{lccc}
  \hline \hline
 & \# selected & \# significant \\ 
  \hline
{\text{MTL (0.7) + SI}} & 89.25 & 45.00 \\
\hline
{\text{DS (0.67)}} & 114.25 & 39.75 \\
{\text{Common}} & 69.50 & 29.25 \\
\hline
{\text{DS (0.5)}} & 78.75 & 35.50 \\
{\text{Common}} &  53.25 &  30.50 \\
   \hline  \hline
\end{tabular}
\caption{Comparison of selected models, in average number of predictors per task}
\label{tab:avg_results}
\end{table}

\begin{figure}[h]
    \centering
    \includegraphics[width=\linewidth]{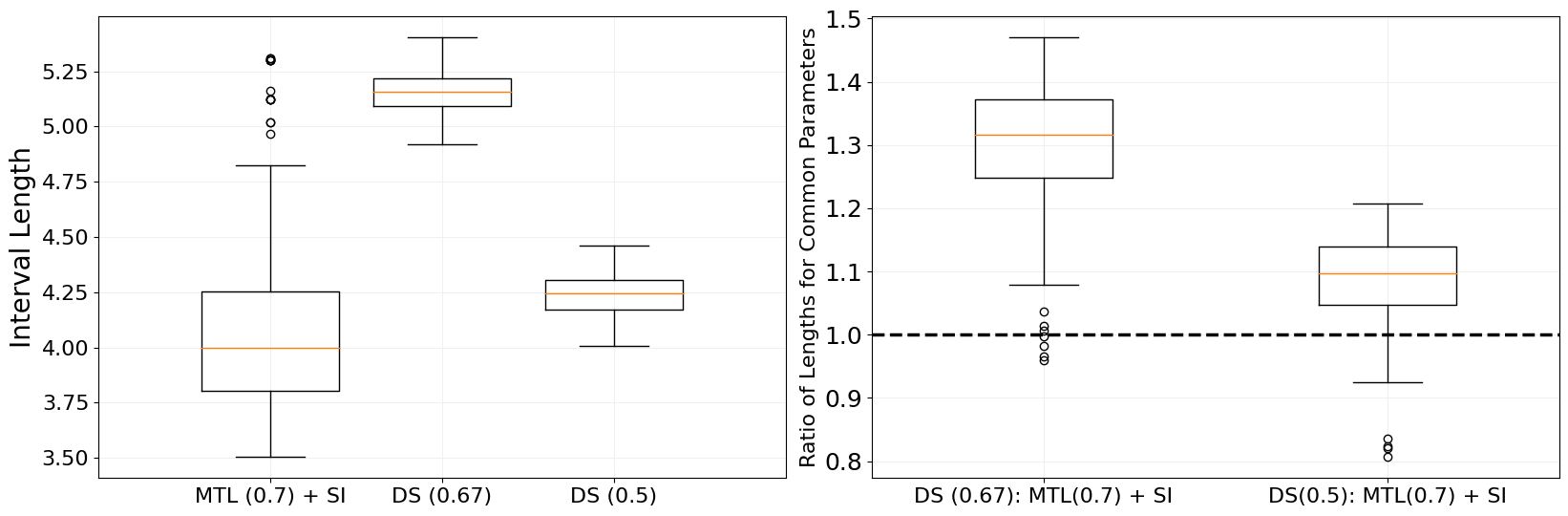}
  \caption{Comparison of inferential power between MTL(0.7) + SI, DS(0.67), and DS(0.5) when applied jointly to all four tasks.  Left: Boxplots of interval lengths.  Right: Ratio of interval lengths for common parameters.}
  \label{fig:inference_results}
\end{figure}

\subsection{Neurological Interpretation of Selected Models}
\label{sec:5.5}

We now offer an interpretation of the shared structure identified through multi-task learning and validated through selective inference. Since experts are ultimately interested in understanding the effects of the connectome, we aim to translate the results back into the original 87,153-dimensional connectome feature space. Note that the mean structure for task $k$ in the original feature space of dimension $p$ can be represented as $$X \beta^{(k)} = X M M'  \beta^{(k)} = Z \theta^{(k)},$$ where $M$ is the orthogonal matrix of eigenvectors for $X'X$, $Z$ is the matrix of principal component scores, and $\theta^{(k)} = M'  \beta^{(k)} \in \mathbb{R}^p$. To recover $\beta^{(k)},$ a potential plug-in estimator is $$\hat{\beta}^{(k)}=M \begin{pmatrix} \hat{\theta}_{E_k}^{{(k)}} \\ 0 \end{pmatrix} =M_{E_k}\hat{\theta}_{E_k}^{(k)},$$ where $\hat{\theta}_{E_k}^{(k)} \in \mathbb{R}^{\delta_k}$ is the MLE for task $k$ obtained from Algorithm \ref{alg:Alginference}.  Although this estimator is only consistent for $\beta^{(k)}$ when $\delta_k=p$, we will use it to approximate $\beta^{(k)}$ as is typically done in ordinary principal components regression \citep{jolliffe2003principal}.

%\liza{The term "network" to mean a region is very confusing for statisticians.  We also never included the list of Power regions, so it's not easy to know what the abbreviations mean.  I think it would be helpful to add a small table listing region abbreviations and full names (under the heading "Brain regions in the Power parcelation"), and then it will be easier.  I have rewritten the below to clarify the use of network but I think adding a table would help further.   Once you add it, please make sure that the first time a network is mentioned, its full name is written out along with the abbreviation, and subsequent mentions can be of either, depending on what sounds better.  }%
Figure \ref{fig:mtl-original} shows the results from applying MTL(0.7) + SI jointly to all four tasks when projected back into the original feature space. There is notable similarity of connectivity patterns across all of the four tasks, providing some support for the general factor model. Key connectivity motifs include stronger positive and negative connections within the default mode network (DMN) and the visual network, as well as increased positive connectivity within the cerebellum. We also observe complex patterns of connectivity changes between regions/networks, especially between the frontoparietal network (FPN) and DMN; the auditory network and the somatomotor network; the cingulo-parietal network (CP) and the retrosplenial network (RSP); and RSP and DMN. 

Our results show some agreement with a recent study that identified FPN, DMN, the dorsal attention network, and the visual network as influential in predicting general intelligence \citep{tong2022transdiagnostic}. FPN is involved in flexible adaptive control \citep{cole2013multi}, and a number of previous studies implicate it in executive functions and cognitive control \citep{cole2007cognitive, niendam2012meta}, constructs closely related to the general factor of intelligence \citep{chen2019testing}. DMN is involved in spontaneous thought \citep{andrews2014default} and semantic/conceptual representation \citep{wirth2011semantic,binder2011neurobiology,binder2009semantic}, capacities that likely facilitate abstraction and problem-solving. Cerebellum has been traditionally associated with coordination of movement \citep{spencer2005role,bastian2006learning}, but there is growing recognition that it coordinates both external motor operations as well as internal mental operations, and thus it plays a critical role in supporting complex cognition \citep{schmahmann2019cerebellum, schmahmann1996movement, andreasen1999defining}. CP and RSP are small networks that were identified relatively recently \citep{gordon2017individual}, and their significance for higher cognitive functions requires further elucidation.

\begin{figure}
    \centering
    \includegraphics[width=\linewidth]{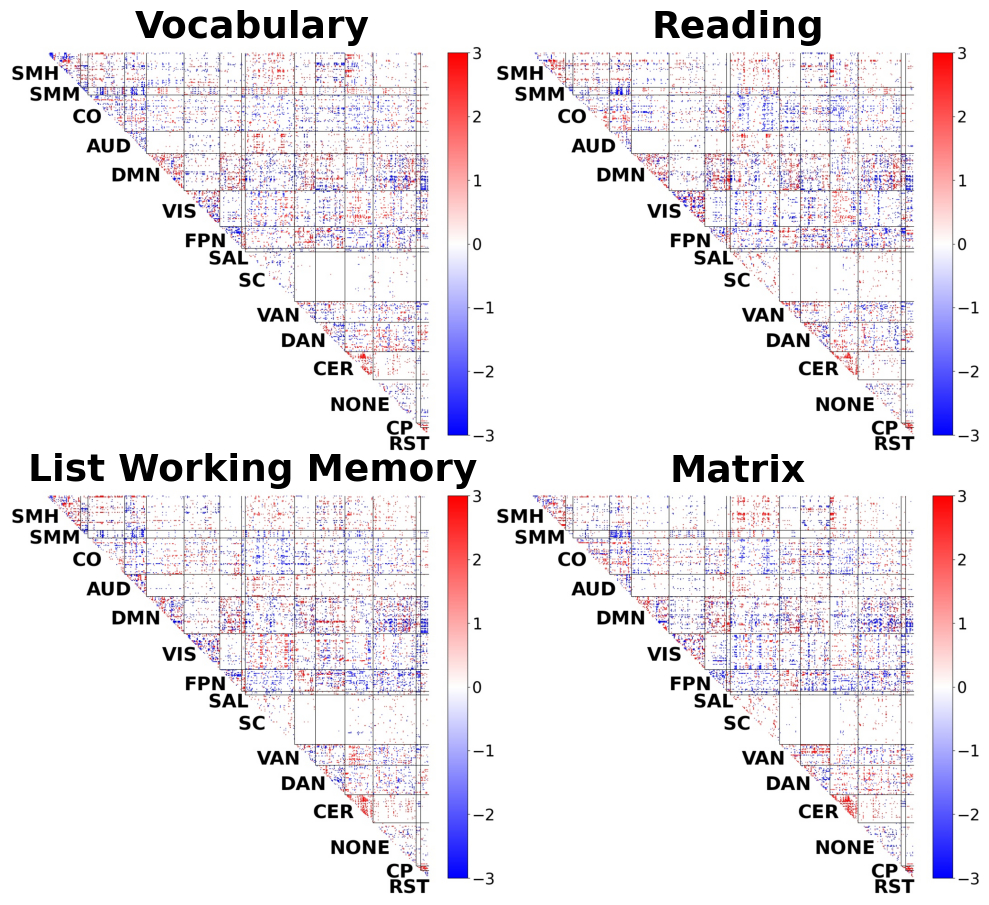}
  \caption{Cartographic visualization showing standardized estimates of the original coefficients, grouped by region/network, that were obtained after applying MTL(0.7) + SI jointly to all four tasks. A threshold of 1.5 is used to improve readability, with smaller-magnitude coefficients set to zero.}
  \label{fig:mtl-original}
\end{figure}

\begin{table}[ht]
\centering
\begin{tabular}{ll | ll}
  \hline \hline
Abbreviation & Full Region/Network Name & Abbreviation & Full Region/Network Name \\ 
\hline
SMH & Somatomotor-Hand & SMM & Somatomotor-Mouth\\
CO & Cingulo-Opercular & AUD & Auditory \\
DMN & Default & VIS & Visual \\
FPN & Frontoparietal & SAL & Salience \\
SC & Subcortical & VAN & Ventral Attention \\
DAN & Dorsal Attention & CER & Cerebellum \\
MONE & Not Named & CP & Cingulo-Parietal \\
RST & Retrosplenial Temporal \\
   \hline  \hline
\end{tabular}
\caption{Names and abbreviations for the 15 groups of ROIs, known as regions/networks}
\label{tab:Networks}
\end{table}

Overall, we find compelling evidence for shared connectivity patterns between the two fluid intelligence tasks and the two crystallized intelligence tasks. A joint approach to model selection recovers more shared structure across all four tasks than performing model selection separately, and selective inference provides additional statistical reassurance that the shared structure we identified through MTL reflects real patterns in the data. Our approach contributes a new perspective to the literature on theories of intelligence, using neurological data to discern the relationships between tasks rather than inferring them from behavioral data.

\section{Discussion}
\label{discussion}

In this paper, we address two key limitations of previous research on the role of the functional connectome in human cognition. Most previous studies have investigated the connectivity patterns associated with general cognitive ability, inferred from behavioral data, or the connectivity patterns associated with specific cognitive abilities like working memory.  By contrast, we leverage a MTL approach to discover the relationship between cognitive domains entirely from shared patterns in the brain-behavior data, avoiding the need to make any inferences about general ability from behavioral data alone. Most prior studies on the role of the functional connectome in human cognition have also focused on predictive power, neglecting statistical inference. We address this shortcoming by developing selective inference procedures to measure the strength and certainty of each brain-behavior relationship discovered from the data, offering improved interpretability. 

By applying our selective inference procedures for MTL to two tasks from the ABCD study that implicate fluid intelligence and two tasks from the ABCD study that implicate crystalized intelligence, we uncover additional shared structure that cannot be detected by modeling the fluid and crystalized tasks separately. This shared structure provides initial support for a general factor model of cognitive abilities. After applying MTL+ SI  to the data, we use cartographic mapping to visualize the brain regions involved in each of the four tasks. Our results reveal that connections involving DMN, FPN, and visual network generally have the most predictive power across tasks, showing some agreement with previous studies. 

\section{Acknowledgments}
We would like to thank Ji Zhu and Daniel Kessler for their helpful feedback throughout the project;  
Qianhua Shan for guidance in accessing the ABCD dataset; Aman Taxali for generating the cartographic maps and helping to run code in parallel on the University of Michigan's high-performance computing cluster; Michael Angstadt for creating the NDA study associated with this paper; and Tian Xie and Qiang Chen for testing different simulations as part of their undergraduate research project.

Data used in the preparation of this article were obtained from the Adolescent Brain Cognitive Development (ABCD) Study (\url{https://abcdstudy.org}), held in the NIMH Data Archive (NDA). This is a multisite, longitudinal study designed to recruit more than 10,000 children age 9-10 and follow them over 10 years into early adulthood. The ABCD Study is supported by the National Institutes of Health and additional federal partners under award numbers U01DA041022, U01DA041028, U01DA041048, U01DA041089, U01DA041106, U01DA041117, U01DA041120, U01DA041134, \\ U01DA041148, U01DA041156, U01DA041174, U24DA041123, and U24DA041147. A full list of supporters is available at \url{https://abcdstudy.org/nih-collaborators}. A listing of participating sites and a complete listing of the study investigators can be found at \url{https://abcdstudy.org/consortium_members/}. ABCD consortium investigators designed and implemented the study and/or provided data but did not necessarily participate in analysis or writing of this report. This manuscript reflects the views of the authors and may not reflect the opinions or views of the NIH or ABCD consortium investigators. The ABCD data repository grows and changes over time. The ABCD data used in this report came from NDA Study 721, 10.15154/1504041, which can be found at \url{https://nda.nih.gov/study.html?id=721}. The specific NDA study associated with this report is NDA Study 1689, 10.15154/1527789.

%%%%%%%%%%%%%%%%%%%%%%%%%%%%%%%%%%%%%%%%%%%%%%
%% Funding information, if any,             %%
%% should be provided in the                %%
%% funding section.                         %%
%%%%%%%%%%%%%%%%%%%%%%%%%%%%%%%%%%%%%%%%%%%%%%

S. Panigrahi's research is supported in part by NSF grants 1951980 and 2113342.
N.  Stewart is supported in part by NSF RTG grant 1646108 and a Rackham Science Award from the University of Michigan. E. Levina's research is supported in part by NSF grants 1916222, 2052918, and 2210439 and NIH grant R01MH123458.

\bibliographystyle{apalike} % Style BST file
\bibliography{references}

\newpage 
\section{Appendix}
\addcontentsline{toc}{section}{Appendices}
\renewcommand{\thesubsection}{\Alph{subsection}}

\subsection{Proofs of Results in Section 3}
\label{App:sec1}
%%%%%%%%%%%%%%%%%%%%%%%%%%%%%%%%%%%%%%
%%% Result 3: New conditioning event and equivalent constraints %%%%%
%%%%%%%%%%%%%%%%%%%%%%%%%%%%%%%%%%%%%%
\begin{proof} Lemma 3.1. \ \  
The proof follows immediately from the definition of $\Stacked{V}$, $\StackedSymbol{\Gamma}$ and by noting that 
the constraints $b^{(k)}>0$ for all $k\in [K]$ are equivalent to
$$ \Stacked{V}>\mathbf{0}, \  \StackedSymbol{\Gamma} - D\Stacked{V} > \mathbf{0}.$$
\end{proof}

%%%%%%%%%%%%%%%%%%%%%%%%%%%%%%%%%%%%%%%%%%%%%%
%%% Result 4: Tractable Conditional likelihood from law of refitted OLS estimates %%%%%
%%%%%%%%%%%%%%%%%%%%%%%%%%%%%%%%%%%%%%%%%%%%%%
\begin{proof} Theorem 3.2. \ \  
Based on the independence between $\h{\StackedSymbol{\beta}}_{\Stacked{E}}$, $\h{\StackedSymbol{\beta}}_{\perp}$ and $\Stacked{W}$, the unconditional law for the corresponding variables is given by
$$ 
 \prod_{k=1}^K \rho\left(\; \h{\beta}_{E_k}^{(k)}; \beta_{E_k}^{(k)}, \sigma_k^2 \cdot \left({X^{(k)}_{E_k}}^{\prime} X^{(k)}_{E_k}\right)^{-1}\right)\rho\left(\h{\beta}_{\perp}^{(k)}; \mathbf{0}, \Sigma^{(k)}_{\perp} \right)\rho\left(w^{(k)}; \mathbf{0}, \Omega^{(k)}\right),
$$
treating $E_k$ for $k\in [K]$ as fixed sets.
Consider the change of variables
\begin{equation}
\label{CoV}
\Stacked{W} \stackrel{\left(\Pi_{\Stacked{X}'\Stacked{Y}}\circ\Psi_A\right)^{-1} }{\xrightarrow{\hspace*{1.25cm}}}  \begin{pmatrix}\Stacked{V} &\StackedSymbol{\Gamma} & \Stacked{U} \end{pmatrix},
\end{equation}
formed by inverting the composition of $\Pi_{\Stacked{X}^{\prime}\Stacked{Y}}$, defined in (11), and $\Psi_{\mathcal{A}}$.
Applying \eqref{CoV}, we observe that that the law of the variables $\h{\StackedSymbol{\beta}}_{\Stacked{E}}$, $\h{\StackedSymbol{\beta}}_{\perp}$, $\Stacked{V}$, $\StackedSymbol{\Gamma}$ and  $\Stacked{U}$ is equal to
\begin{equation*}
\begin{aligned}
& \big\lvert{J_{\left(\Pi_{\Stacked{X}^{\prime}\Stacked{Y}}\circ\Psi_A\right)}}\big\vert \left\{\prod_{k=1}^K  \rho\left(\h{\beta}_{E_k}^{(k)}; \beta_{E_k}^{(k)}, \sigma_k^2 \cdot \left({X^{(k)}_{E_k}}^{\prime} X^{(k)}_{E_k}\right)^{-1}\right)\rho\left(\h{\beta}_{\perp}^{(k)}; \mathbf{0}, \Sigma^{(k)}_{\perp}\right) \right\}  \ \\
&\;\;\;\;\;\;\;\;\;\;\;\;\;\;\;\;\;\;\;\;\;\;\;\;\;\;\;\;\;\;\;\;\;\;\;\;\;\;\;\;\;\;\;\;\;\;\;\;\;\;\;\;\;\;\;\;\;\;\;\;\;\;\;\;\;\;\;\;\;\;\;\;\;\;\;\;\;\;\;\;\;\;\times\rho\left(\Pi_{\Stacked{X}^{\prime}\Stacked{Y}}\circ\Psi_A \left( \Stacked{V}, \StackedSymbol{\Gamma}, \Stacked{U} \right); \mathbf{0}, \bm{\Omega}\right),
\end{aligned}
\end{equation*}
where we let $J_{\left(\Pi_{\Stacked{X}^{\prime}\Stacked{Y}}\circ\Psi_A\right)} $ be the Jacobian for the change of variables mapping and $\Stacked{\Omega}\in \mathbb{R}^{Kp\times Kp}$ is the block diagonal matrix with the diagonal entries equal to $\Omega^{(k)} \in \mathbb{R}^{p\times p}$ for $k \in [K]$.
To complete the details of the proof, we introduce some matrices.
Let $\mathcal{A}= \begin{bmatrix} \mathcal{A}_1 & \mathcal{A}_2 \end{bmatrix}$. Note that
$$\left(\Pi_{\Stacked{X}^{\prime}\Stacked{Y}}\circ\Psi_A\right) \left( \Stacked{V}, \StackedSymbol{\Gamma}, \Stacked{U} \right) = \mathcal{C}_1 \h{\StackedSymbol{\beta}}_{\Stacked{E}} + \mathcal{C}_2 \Stacked{V} + f\left( \StackedSymbol{\Gamma}, \Stacked{U} ; \h{\StackedSymbol{\beta}}_{\perp} \right) $$
where $$\mathcal{C}_1 = - \Dg \left( {X^{(1)}}'{X^{(1)}_{E_1}},\dots,{X^{(K)}}'{X^{(K)}_{E_K}}\right) $$
$$\mathcal{C}_0 = \left(- \mathcal{C}_1 + \epsilon \; \Dg \left( \left(\rm{I}_{\delta_1} \;\;\; \mathbf{0}  \right)',\dots,\left(\rm{I}_{\delta_K} \;\;\; \mathbf{0}  \right)' \right) \right) \Dg \left(\Stacked{S}\right)$$
$$\mathcal{C}_2 = \mathcal{C}_0 \left(\mathcal{A}_1 - \mathcal{A}_2 D \right), \text{ and }$$
$$f\left( \StackedSymbol{\Gamma}, \Stacked{U} ; \h{\StackedSymbol{\beta}}_{\perp} \right) =  \Dg\left(\Stacked{\Lambda}\right) \left( \Stacked{S}' \;\; \Stacked{U}' \right)' + \mathcal{C}_0 \mathcal{A}_2 \StackedSymbol{\Gamma} - \h{\StackedSymbol{\beta}}_{\perp}.$$
Additionally, we define 
$$\Delta^{-1} = \mathcal{C}_2' \StackedSymbol{\Omega}^{-1} \mathcal{C}_2; \;\;\; P= -\Delta \mathcal{C}_2' \StackedSymbol{\Omega}^{-1} \mathcal{C}_1; \;\;\; q = -\Delta \mathcal{C}_2' \StackedSymbol{\Omega}^{-1} f \left( \StackedSymbol{\Gamma}, \Stacked{U} ; \h{\StackedSymbol{\beta}}_{\perp} \right);$$
$$\Sigma^{-1} =  \Dg \left( \frac{1}{\sigma_1^2} {X^{(1)}_{E_1}}'{X^{(1)}_{E_1}}, \dots, \frac{1}{\sigma_K^2} {X^{(K)}_{E_K}}'{X^{(K)}_{E_K}} \right) + \mathcal{C}_1' \StackedSymbol{\Omega}^{-1} \mathcal{C}_1 - P' \Delta^{-1} P;$$
$$L = \Sigma \; \Dg \left( \frac{1}{\sigma_1^2} {X^{(1)}_{E_1}}'{X^{(1)}_{E_1}}, \dots, \frac{1}{\sigma_K^2} {X^{(K)}_{E_K}}'{X^{(K)}_{E_K}} \right); \text{ and }$$
$$m = \Sigma \; P' \Delta^{-1} q-\Sigma \; \mathcal{C}_1' \StackedSymbol{\Omega}^{-1} f\left( \StackedSymbol{\Gamma}, \Stacked{U} ; \h{\StackedSymbol{\beta}}_{\perp} \right).$$
Conditioning upon $\h{\StackedSymbol{\beta}}_{\perp}$ and the event in (12), which we note is equivalent to 
$$ \left\{ \Stacked{V}>\mathbf{0}, \  \StackedSymbol{\Gamma} - D \Stacked{V} > 0 , \ \h{\StackedSymbol{\Gamma}} = \StackedSymbol{\Gamma}, \ \h{\Stacked{U}}=\Stacked{U} \right\}$$
based on Lemma 3.1, the law for 
$\h{\StackedSymbol{\beta}}_{\Stacked{E}}$ and $\Stacked{V}$ is given up to a constant by
\begin{equation*}
\begin{aligned}
&\left(\int \rho\left( \; \localvar{\beta}; L \StackedSymbol{\beta}_{\Stacked{E}} +m, \Sigma \right) \cdot \rho(\localvar{V}; P\localvar{\beta} + q, \Delta) \cdot 1\left(H \localvar{V} \geq g \right) \ d\localvar{V} d \localvar{\beta} \; \right)^{-1} \cdot \\ & \;\;\;\;\;\;\;\;\;\;\;\;\;\;\;\;\;\;\;\;\;\;\;\;\;\;\;\;\;\;\;\;\;\;\;\;\;\; \rho\left(\h{\StackedSymbol{\beta}}_{\Stacked{E}}; L\StackedSymbol{\beta}_{\Stacked{E}} +m, \Sigma \right) \cdot \rho(\Stacked{V}; P \h{\StackedSymbol{\beta}}_{\Stacked{E}} + q, \Delta)  \cdot 1\left(H \Stacked{V} \geq g \right).
\end{aligned}
\end{equation*}
Integrating over $\Stacked{V}$ in the last display and noting that 
$$\int_{\{\localvar{V}: H\localvar{V} \geq g \}}\rho(\localvar{V}; P\h{\StackedSymbol{\beta}}_{\Stacked{E}} + q, \Delta)d \localvar{V}$$
is a constant free from the parameters $\StackedSymbol{\beta}_{\Stacked{E}}$ in our model, we complete the proof of the Theorem.
\end{proof}

%%%%%%%%%%%%%%%%%%%%%%%%%%%%%%%%%%%%
%%% Result 5: Estimating Equations for the MLE and Fisher Info %%%
%%%%%%%%%%%%%%%%%%%%%%%%%%%%%%%%%%%%
\begin{proof} Theorem 3.3. \ \  
Observe that the \
maximum likelihood estimate is obtained by maximizing the following objective with respect to $\StackedSymbol{\beta}_{\Stacked{E}}$:
\begin{equation}
\label{mle:problem}
\begin{aligned}
& \underset{\StackedSymbol{\beta}_{\Stacked{E}}}{\text{maximize}}\ \  \h{\StackedSymbol{\beta}}_{\Stacked{E}}^{\prime}\Sigma^{-1}(L\StackedSymbol{\beta}_{\Stacked{E}} + m) +  \underset{\localvar{\beta}, \localvar{V}}{\inf}\; \Big\{ \dfrac{1}{2} \localvar{\beta}' \Sigma^{-1} \localvar{\beta} -\localvar{\beta}'\Sigma^{-1} (L\StackedSymbol{\beta}_{\Stacked{E}} + m)\\
&\;\;\;\;\;\;\;\;\;\; \;\;\;\;\;\;\;\;\;\; \;\;\;\;\;\;\;\;\;\; \;\;\;\;\;\;\;\;\;\; \;\;\;\;\;\;\;\;\;\; + \dfrac{1}{2}(\localvar{V}-P\localvar{\beta} - q)'\Delta^{-1} (\localvar{V}-P\localvar{\beta} - q)+ \phi_{H,g}(\localvar{V})\Big\}.
\end{aligned}
\end{equation}
Let $$F^*(\zeta) = \underset{\localvar{\beta}}{\sup}\; \Big\{\localvar{\beta}'\zeta - \dfrac{1}{2} \localvar{\beta}' \Sigma^{-1} \localvar{\beta} -\underset{\localvar{V}}{\inf} \; \dfrac{1}{2}(\localvar{V}-P\localvar{\beta} - q)' \Delta^{-1} (\localvar{V}-P\localvar{\beta} - q)+ \phi_{H,g}(\localvar{V})\Big\}$$
represent the convex conjugate for
$$F(\beta)= \dfrac{1}{2} \beta' \Sigma^{-1} \beta + \underset{\localvar{V}}{\inf} \; \dfrac{1}{2}(\localvar{V}-P\beta - q)'\Delta^{-1} (\localvar{V}-P\beta - q)+ \phi_{H,g}(\localvar{V}).$$
Then, solving \eqref{mle:problem} is equivalent to the problem 
\begin{equation*}
\begin{aligned}
& \underset{\StackedSymbol{\beta}_{\Stacked{E}}}{\text{maximize}}\ \  \h{\StackedSymbol{\beta}}_{\Stacked{E}}'\Sigma^{-1} (L\StackedSymbol{\beta}_{\Stacked{E}} + m) -  \underset{\localvar{\beta}}{\sup}\; \Big\{\localvar{\beta}'\Sigma^{-1} (L\StackedSymbol{\beta}_{\Stacked{E}} + m) - \dfrac{1}{2} \localvar{\beta}' \Sigma^{-1} \localvar{\beta}\\
&\;\;\;\;\;\;\;\;\;\; \;\;\;\;\;\;\;\;\;\; \;\;\;\;\;\;\;\;\;\; \;\;\;\;\;\;\;\;\;\; \;\;\;\;\;\;\;\;\;\;  -\underset{\localvar{V}}{\inf}\; \dfrac{1}{2}(\localvar{V}-P \localvar{\beta} - q)^{\prime}\Delta^{-1} (\localvar{V}-P\localvar{\beta} - q)+ \phi_{H,g}(\localvar{V})\Big\}\\
&= \underset{\StackedSymbol{\beta}_{\Stacked{E}}}{\text{maximize}}\ \  \h{\StackedSymbol{\beta}}_{\Stacked{E}}'\Sigma^{-1}(L\StackedSymbol{\beta}_{\Stacked{E}} + m) - F^*(\Sigma^{-1} (L\StackedSymbol{\beta}_{\Stacked{E}} + m)).
\end{aligned}
\end{equation*}
Using the identity $(\grad  F^*)^{-1} = \grad F$ from convex analysis, we have 
\begin{equation*}
\begin{aligned}
\Sigma^{-1}(L\Mb_{\Stacked{E}} + m) &= (\grad  F^*)^{-1}(\h{\StackedSymbol{\beta}}_{\Stacked{E}})\\
&= \Sigma^{-1}\h{\StackedSymbol{\beta}}_{\Stacked{E}} + \Big(\dfrac{\partial}{\partial \beta}V(\beta)\Big\lvert_{\beta=\h{\StackedSymbol{\beta}}_{\Stacked{E}}} -P \Big)' \Delta^{-1} (\h{\Stacked{V}}-P\h{\StackedSymbol{\beta}}_{\Stacked{E}} - q) \\ &\;\;\;\;\;\;\;\;\;\;\;\;\;\;\;\;\;\;\;\;\;\;\;\;\;\;\;\;\;\;\;\;\;\;\;\;\;\;\;\;\;\;\;\;\;\;\;\;\;\;\;\;\; + \Big(\dfrac{\partial}{\partial \beta}V(\beta)\Big\lvert_{\beta=\h{\StackedSymbol{\beta}}_{\Stacked{E}}}\Big)' \grad\phi_{H,g}(\h{\Stacked{V}})
\end{aligned},
\end{equation*}
where 
%$\h{\Stacked{V}}$ is the optimizer defined in \eqref{eq:opt} and
\begin{equation}
\label{V:opt}
V(\beta) = \underset{\localvar{V}}{\text{argmin}} \; \frac{1}{2}(\localvar{V}-P\beta - q)' \Delta^{-1} (\localvar{V}-P\beta - q) + \phi_{H,g}(\localvar{V})
\end{equation}
and $V(\h{\StackedSymbol{\beta}}_{\Stacked{E}}) = \h{\Stacked{V}}$, the optimizer defined in  (16).
By the definition of $\h{\Stacked{V}}$,
$$ \Delta^{-1} (\h{\Stacked{V}}-P\h{\StackedSymbol{\beta}}_{\Stacked{E}} - q) + \grad\phi_{H,g}(\h{\Stacked{V}}) = 0,$$
and we derive the system of estimating equations
$$\Mb_{\Stacked{E}} = L^{-1}\h{\StackedSymbol{\beta}}_{\Stacked{E}} + L^{-1}\Sigma P' \Delta^{-1}(P\h{\StackedSymbol{\beta}}_{\Stacked{E}} + q - \h{\Stacked{V}})-L^{-1}m.$$

\medskip
\noindent Using the curvature of the smooth approximate likelihood, we observe that the expression for the observed Fisher information matrix is equal to
\begin{equation*}
\begin{aligned}
L' \Sigma^{-1} \grad^2 F^*(\zeta)\lvert_{\zeta= \Sigma^{-1}(L\Mb_{\Stacked{E}} + m)} \Sigma^{-1}L &= 
&=  L' \Sigma^{-1}\dfrac{\partial}{\partial\zeta} \beta^*(\zeta)\Big\lvert_{\zeta= \Sigma^{-1}(L\Mb_{\Stacked{E}} + m)} \Sigma^{-1}L
\end{aligned}.
\end{equation*}
Here, $\grad^2$ denotes the Hessian matrix, and 
\begin{equation}
\label{b:opt}
\beta^*(\zeta) = \underset{\localvar{\beta}}{\text{argsup}}\; \Big\{\localvar{\beta}'\zeta - \dfrac{1}{2} \localvar{\beta}' \Sigma^{-1} \localvar{\beta} -\underset{\localvar{V}}{\inf} \; \dfrac{1}{2}(\localvar{V}-P\localvar{\beta} - q)^{\prime}\Delta^{-1} (\localvar{V}-P\localvar{\beta} - q)+ \phi_{H,g}(\localvar{V})\Big\}.
\end{equation}
From the K.K.T. conditions of optimality for \eqref{b:opt},
\begin{equation*}
\zeta =  (\Sigma^{-1}+  P' \Delta^{-1} P)\beta^*(\zeta) -  P' \Delta^{-1}(V(\beta^*(\zeta)) - q),
\end{equation*}
which leads to noting
$$\dfrac{\partial}{\partial\beta} V(\beta)\Big\lvert_{\beta= \beta^*(\zeta)}= \left(\Delta^{-1} + \grad^2 \phi_{H,g}({V(\beta^*(\zeta)) })\right)^{-1}\Delta^{-1}P,$$
and 
\begin{equation*}
\begin{aligned}
\dfrac{\partial}{\partial\zeta} \beta^*(\zeta)\Big\lvert_{\zeta= \zeta_0} &= \left( \Sigma^{-1} + P' \Delta^{-1}P - P' \Delta^{-1}\dfrac{\partial}{\partial\beta} V(\beta)\Big\lvert_{\beta= \beta^*(\zeta)} \right )^{-1}\\
&= \left( \Sigma^{-1} +P' \Delta^{-1}P - P' \Delta^{-1}\left(\Delta^{-1} + \grad^2 \phi_{H,g}(V(\beta^*(\zeta_0)))\right)^{-1}\Delta^{-1}P \right )^{-1}.
\end{aligned}
\end{equation*}
Using, once again, the K.K.T. conditions of optimality for \eqref{b:opt}, we observe that $$\beta^*\left(\Sigma^{-1}(L\Mb_{\Stacked{E}} + m)\right)= \h{\StackedSymbol{\beta}}_{\Stacked{E}}.$$
Plugging 
$$\dfrac{\partial}{\partial\zeta} \beta^*(\zeta)\Big\lvert_{\zeta= \Sigma^{-1}(L\Mb_{\Stacked{E}} + m)}$$
into the expression for the observed Fisher information, we obtain the expression for the inverse information matrix:
$$\InvFI=L^{-1} \Sigma {L'}^{-1}  + L^{-1} \Sigma\left(P' \Delta^{-1}P - P' \Delta^{-1} \left(\Delta^{-1} + \grad^2 \phi_{H,g}(\h{\Stacked{V}}) \right)^{-1}\Delta^{-1}P \right)  \Sigma {L'}^{-1}.$$
\end{proof}

\subsection{Impracticality of Existing Conditional Prescription in Multi-Task Setting}
\label{App:sec2}

In this section, we show in detail why the existing conditional prescription for other $\ell_1$-regularized algorithms after conditioning on the event in (10) does not lead to tractable inference. Given a stationary point $\left(\Stacked{B}, \Stacked{U} \right)$ defined through the K.K.T. mapping in (11), we immediately note the following equivalence
\begin{equation}
\label{Lee:event:characterization}
\h{\Stacked{E}} = \Stacked{E}, \;\h{\Stacked{S}}=\Stacked{S}\;\; \equiv \;\; b^{(k)} >\mathbf{0}, \|u^{(k)}\|_{\infty}\leq 1 \text{ for all } k \in [K].
\end{equation}
The alternative characterization of the selection event through the K.K.T. mapping is directly analogous to the characterization used previously in conducting selective inference for the single-task LASSO. From here, the stationary mapping suggests a change of variables that may be useful in deriving a joint density function for the response and randomization variables. Proposition \ref{prop:jacobian} provides us an expression for the Jacobian associated with the change of variables, and Proposition \ref{prop:Leelaw} describes the resulting selection-adjusted likelihood using the marginal law of the response vectors for the $K$ tasks. Let
$Q = \Dg(Q^{(1,1)},\dots,Q^{(K,K)})$, where the block matrices $Q^{(k,k)} \in R^{p-\delta_k \times p-\delta_k}$ have columns given by $$
 Q_{\cdot,\widetilde{j}_{k}}^{(k,k)} =  \Lambda_{\widetilde{j}_k}^{(k)} \cdot  \mathbf{e}_{\ \widetilde{j}_k}\;\;\;\text{ for } j \in -E_k$$
Here, $\mathbf{e}_{\ \widetilde{j}_k} \in \mathbb{R}^{p-\delta_k}$ is the standard basis vector whose ${\widetilde{j}_k}^{\text{th}}$ component is one, and all other components are zero. Let $$R = \begin{pmatrix} R^{(1,1)} & \dots & R^{(1,K)} \\ \vdots & \ddots & \vdots \\
R^{(K,1)} & \dots & R^{(K,K)} \\ \end{pmatrix} ,$$ where the block matrices  $R^{(k,k')} \in \mathbb{R}^{\delta_k \times \delta_{k'}}$ are defined column-wise as
$$R^{(k,k')}_{\cdot , \widetilde{j}_{k'}} = \begin{cases}   -\left(\dfrac{1}{2} \lambda \hspace{1mm}  s^{(k)}_{\ \widetilde{j}_k} \hspace{1.5mm} {\Gamma^{(j)}}^{-\frac{3}{2}} s^{(k')}_{\ \widetilde{j}_{k'}}\right) \cdot \mathbf{e}_{\ \widetilde{j}_{k}}  & \text{ if } j \in E_k \text{ and } \Lambda^{(k)}_{\widetilde{j}_k} > \lambda_0 \\ \ \  \ \  \mathbf{0} & \text{ otherwise}\end{cases} \;\;\;\;\;\text{ for } j \in E_{k'}.$$
Here, $\mathbf{e}_{\ \widetilde{j}_k} \in \mathbb{R}^{\delta_k}$ is the standard basis vector whose ${\widetilde{j}_k}^{\text{th}}$ component is one, and all other components are zero. Finally, let $$ T = \Dg\begin{pmatrix} {({X^{(1)}_{E_1}})}' X^{(1)}_{E_1} + \epsilon \cdot \rm{I}_{\delta_1} & {({X^{(2)}_{E_2}})}' X^{(2)}_{E_2} + \epsilon \cdot \rm{I}_{\delta_2} & \dots & {({X^{(K)}_{E_K}})}' X^{(K)}_{E_K} + \epsilon \cdot \rm{I}_{\delta_K} \end{pmatrix} \cdot \Dg(\mathbf{S}).$$
\begin{proposition}
\label{prop:jacobian}
Consider a change of variables $\Stacked{W} \stackrel{\Pi_{\Stacked{X}'\Stacked{Y}}^{-1}}{\longrightarrow}  \begin{pmatrix}\Stacked{B} & \Stacked{U} \end{pmatrix}$ where
 $$\Stacked{W} = \Pi_{\Stacked{X}'\Stacked{Y}}\begin{pmatrix}\Stacked{B} & \Stacked{U} \end{pmatrix} := \begin{pmatrix} \pi^{(1)}(b^{(1)},u^{(1)}) & \pi^{(2)}(b^{(2)},u^{(2)})  &\cdots & \pi^{(K)}(b^{(K)},u^{(K)})\end{pmatrix}$$
and such that $\pi^{(k)}(\cdot)$ is defined according to (11).
The Jacobian determinant associated with the above change of variables is given by
\begin{align*}
\big|J_{\Pi_{\Stacked{X}'\Stacked{Y}}} \left( \Stacked{B}, \Stacked{U} \right) \big|
&=  \det \left( Q \right) \det \left( R + T \right).
\end{align*}
\end{proposition}

\begin{proof} \ \
We note that the Jacobian for the change of variables mapping $J_{\Pi_{\Stacked{X}'\Stacked{Y}}}$ is the determinant in absolute value for the following matrix
\begin{align*}
&\begin{bmatrix} \frac{\partial \pi^{(1)}}{\partial (b^{(1)}, u^{(1)})} & \frac{\partial \pi^{(1)}}{\partial (b^{(2)}, u^{(2)})} &  \dots  & \frac{\partial \pi^{(1)}}{\partial (b^{(K)}, u^{(K)})} \\
 \frac{\partial \pi^{(2)}}{\partial (b^{(1)}, u^{(1)})} & \frac{\partial \pi^{(2)}}{\partial (b^{(2)}, u^{(2)})} &  \dots  & \frac{\partial \pi^{(2)}}{\partial (b^{(K)}, u^{(K)})}\\
 \vdots & \ddots  & & \vdots \\ 
 \frac{\partial \pi^{(K)}}{\partial (b^{(1)}, u^{(1)})} & \frac{\partial \pi^{(2)}}{\partial (b^{(2)}, u^{(2)})} &  \dots  & \frac{\partial \pi^{(K)}}{\partial (b^{(K)}, u^{(K)})} \end{bmatrix} \\
 &= \begin{bmatrix} \frac{\partial \omega_{E_1}^{(1)}}{\partial b^{(1)}} &  \frac{\partial \omega_{E_1}^{(1)}}{\partial u^{(1)}} & \frac{\partial \omega_{E_1}^{(1)}}{\partial b^{(2)}} &  \frac{\partial \omega_{E_1}^{(1)}}{\partial u^{(2)}}  &  \dots  & \frac{\partial \omega_{E_1}^{(1)}}{\partial b^{(K)}} &  \frac{\partial \omega_{E_1}^{(1)}}{\partial u^{(K)}}  \\
 \frac{\partial \omega_{-E_1}^{(1)}}{\partial b^{(1)}} &  \frac{\partial \omega_{-E_1}^{(1)}}{\partial u^{(1)}} & \frac{\partial \omega_{-E_1}^{(1)}}{\partial b^{(2)}} &  \frac{\partial \omega_{-E_1}^{(1)}}{\partial u^{(2)}}  &  \dots  & \frac{\partial \omega_{-E_1}^{(1)}}{\partial b^{(K)}} &  \frac{\partial \omega_{-E_1}^{(1)}}{\partial u^{(K)}}  \\
 \vdots & & & \ddots & & & \vdots \\
 \frac{\partial \omega_{E_K}^{(K)}}{\partial b^{(1)}} &  \frac{\partial \omega_{E_K}^{(K)}}{\partial u^{(1)}} & \frac{\partial \omega_{E_K}^{(K)}}{\partial b^{(2)}} &  \frac{\partial \omega_{E_K}^{(K)}}{\partial u^{(2)}}  &  \dots  & \frac{\partial \omega_{E_K}^{(1)}}{\partial b^{(K)}} &  \frac{\partial \omega_{E_K}^{(K)}}{\partial u^{(K)}}  \\
 \frac{\partial \omega_{-E_K}^{(K)}}{\partial b^{(1)}} &  \frac{\partial \omega_{-E_K}^{(K)}}{\partial u^{(1)}} & \frac{\partial \omega_{-E_K}^{(K)}}{\partial b^{(2)}} &  \frac{\partial \omega_{-E_K}^{(K)}}{\partial u^{(2)}}  &  \dots  & \frac{\partial \omega_{-E_K}^{(K)}}{\partial b^{(K)}} &  \frac{\partial \omega_{-E_K}^{(K)}}{\partial u^{(K)}} \end{bmatrix}.
\end{align*}
Noting in the above display $\frac{\partial \omega_{E_k}^{(k)}}{\partial u^{(k)}}=0$ and $\frac{\partial \omega_{E_k}^{(k)}}{\partial u^{(k')}}=0$ and $\frac{\partial \omega_{-E_k}^{(k)}}{\partial u^{(k')}}=0$ for all $k'\neq k$,
the above Jacobian equals
\begin{align}
\label{Jacob:matrix}
&\begin{bmatrix} \frac{\partial \omega_{E_1}^{(1)}}{\partial b^{(1)}} &  0 & \frac{\partial \omega_{E_1}^{(1)}}{\partial b^{(2)}} &  0  &  \dots  & \frac{\partial \omega_{E_1}^{(1)}}{\partial b^{(K)}} & 0  \\
 \frac{\partial \omega_{-E_1}^{(1)}}{\partial b^{(1)}} &  \frac{\partial \omega_{-E_1}^{(1)}}{\partial u^{(1)}} & \frac{\partial \omega_{-E_1}^{(1)}}{\partial b^{(2)}} &  0  &  \dots  & \frac{\partial \omega_{-E_1}^{(1)}}{\partial b^{(K)}} &  0 \\
 \vdots & & & \ddots & & & \vdots \\
 \frac{\partial \omega_{E_K}^{(K)}}{\partial b^{(1)}} & 0 & \frac{\partial \omega_{E_K}^{(K)}}{\partial b^{(2)}} &  0  &  \dots  & \frac{\partial \omega_{E_K}^{(1)}}{\partial b^{(K)}} &  0 \\
 \frac{\partial \omega_{-E_K}^{(K)}}{\partial b^{(1)}} & 0 & \frac{\partial \omega_{-E_K}^{(K)}}{\partial b^{(2)}} &  0 &  \dots  & \frac{\partial \omega_{-E_K}^{(K)}}{\partial b^{(K)}} &  \frac{\partial \omega_{-E_K}^{(K)}}{\partial u^{(K)}} \end{bmatrix}
\end{align}
which further equals in absolute value the determinant of 
$\begin{bmatrix} 
 M_1 & \mathbf{0}\\
 M_2 & Q
\end{bmatrix}$
where 
$$
M_1= \begin{bmatrix} \frac{\partial \omega_{E_1}^{(1)}}{\partial b^{(1)}} &  \frac{\partial \omega_{E_1}^{(1)}}{\partial b^{(2)}} & \dots & \frac{\partial \omega_{E_1}^{(1)}}{\partial b^{(K)}}\\
\frac{\partial \omega_{E_2}^{(2)}}{\partial b^{(1)}} &  \frac{\partial \omega_{E_2}^{(2)}}{\partial b^{(2)}} & \dots & \frac{\partial \omega_{E_2}^{(2)}}{\partial b^{(K)}} \\
\vdots &  \ddots & & \vdots\\
\frac{\partial \omega_{E_K}^{(K)}}{\partial b^{(1)}} &  \frac{\partial \omega_{E_K}^{(K)}}{\partial b^{(2)}} & \dots & \frac{\partial \omega_{E_K}^{(K)}}{\partial b^{(K)}} 
\end{bmatrix}, 
M_2= \begin{bmatrix} \frac{\partial \omega_{-E_1}^{(1)}}{\partial b^{(1)}} &  \frac{\partial \omega_{-E_1}^{(1)}}{\partial b^{(2)}} & \dots & \frac{\partial \omega_{-E_1}^{(1)}}{\partial b^{(K)}}\\
\frac{\partial \omega_{-E_2}^{(2)}}{\partial b^{(1)}} &  \frac{\partial \omega_{-E_2}^{(2)}}{\partial b^{(2)}} & \dots & \frac{\partial \omega_{-E_2}^{(2)}}{\partial b^{(K)}} \\
\vdots &  \ddots & & \vdots\\
\frac{\partial \omega_{-E_K}^{(K)}}{\partial b^{(1)}} &  \frac{\partial \omega_{-E_K}^{(K)}}{\partial b^{(2)}} & \dots & \frac{\partial \omega_{-E_K}^{(K)}}{\partial b^{(K)}} 
\end{bmatrix}, \text{ and }
$$
$$
Q =\Dg\left({\tfrac{\partial \omega_{-E_1}^{(1)}}{\partial u^{(1)}},  \tfrac{\partial \omega_{-E_2}^{(2)}}{\partial u^{(2)}}, \cdots, \tfrac{\partial \omega_{-E_K}^{(K)}}{\partial u^{(K)}}}\right).
$$
Observe $M_1 = R + T$, because
$$\frac{\partial \omega_{E_k}^{(k)}}{\partial b^{(k)}} = {X^{(k)}_{E_k}}' X^{(k)}_{E_k} + \epsilon \cdot \rm{I}_{\delta_k}  + R^{(k,k)},\  \;  \frac{\partial \omega_{E_k}^{(k)}}{\partial b^{(k')}} = R^{(k,k')} \text{ for } k\neq k'.$$
Using the block diagonal structure of \eqref{Jacob:matrix}, it follows that 
$\big| J_{\Pi_{\Stacked{X}'\Stacked{Y}}} \left( \Stacked{B}, \Stacked{U} \right) \big| = \det \left( Q \right) \det \left( R + T \right)$.
\end{proof}

\begin{proposition}
\label{prop:Leelaw}
Suppose we observe the event (10) after solving (7). Fixing the set
$$\mathcal{O}=\left\{ (\localvar{Y}, \localvar{B}, \localvar{U}): \localvar{b}^{(k)} >\mathbf{0}, \|\localvar{u}^{(k)}\|_{\infty}\leq 1 \text{ for all } k \in [K]\right\},$$
we let
\begin{equation*}
\begin{aligned}
\mathcal{N}_{\mathcal{O}}(\StackedSymbol{\beta}_{\Stacked{E}}) &:= \int_{\mathcal{O}}\Big\lvert{ J_{\Pi_{\Stacked{X}'\localvar{Y}}} \left(\localvar{B}, \localvar{U} \right)}\Big\rvert \cdot \prod_{k=1}^K \rho\left(\localvar{y}^{(k)} ; \, X_{E_k}^{(k)}\beta_{E_k}^{(k)}, \sigma_k^2 \cdot \rm{I}_{n_k} \right) \\
&\;\;\;\;\;\;\;\;\;\;\;\times \exp\left(-\frac{1}{2}\sum_{k=1}^K (\pi^{(k)}(\localvar{b}^{(k)},\localvar{u}^{(k)}))^{\prime} {\Omega^{(k)}}^{-1} \pi^{(k)}(\localvar{b}^{(k)},\localvar{u}^{(k)})\right)d\localvar{B} d\localvar{U} d\localvar{Y}.
\end{aligned}
\end{equation*}
Under the model in (9), the likelihood obtained from the law of $\Stacked{Y}$ conditional upon (10) is given up to a constant by
\begin{equation*}
\left(\mathcal{N}_{\mathcal{O}}(\StackedSymbol{\beta}_{\Stacked{E}})\right)^{-1}\prod_{k=1}^K \rho\left(y^{(k)} ; \, X_{E_k}^{(k)}\beta_{E_k}^{(k)}, \sigma_k^2 \cdot \rm{I}_{n_k} \right).
\end{equation*}
\end{proposition}

\begin{proof} \ \  
Let $\mathcal{L}$ be the set of realizations of the responses and randomization variables that result in observing (10), that is
$$\mathcal{L} =\left\{ (\Stacked{Y}, \Stacked{W}): \h{\Stacked{E}}(\Stacked{Y}, \Stacked{W}) = \Stacked{E}, \; \h{\Stacked{S}}(\Stacked{Y}, \Stacked{W}) =\Stacked{S}\right\}.$$
We note that the joint law of $\Stacked{Y}$, $\Stacked{W}$ after conditioning out this observed event from (10)
is obtained by truncating their unconditional Gaussian law and is given by:
\begin{equation*}
\begin{aligned}
&\dfrac{\left\{\prod_{k=1}^K \rho\left(y^{(k)} ; X_{E_k}^{(k)}\beta_{E_k}^{(k)}, \sigma_k^2 \cdot \rm{I}_{n_k} \right) \cdot \rho(\omega^{(k)}; \mathbf{0}, \Omega^{(k)}) \right\}\cdot 1_{\mathcal{L}}(\Stacked{Y}, \Stacked{W})}{\int \left\{\prod_{k=1}^K \rho\left(\localvar{y}^{(k)}; X_{E_k}^{(k)}\beta_{E_k}^{(k)}, \sigma_k^2 \cdot \rm{I}_{n_k} \right) \cdot \rho(\localvar{w}^{(k)}; \mathbf{0}, \Omega^{(k)}) \right\}\cdot 1_{\mathcal{L}}(\localvar{Y}, \localvar{W})d\localvar{W} d\localvar{Y}}.
\end{aligned}
\end{equation*}
\noindent Using the change of variables mapping 
$$\Stacked{W} = \Pi_{\Stacked{X}'\Stacked{Y}}\begin{pmatrix}\Stacked{B} & \Stacked{U} \end{pmatrix} = \begin{pmatrix} \pi^{(1)}(b^{(1)},u^{(1)}) & \pi^{(2)}(b^{(2)}, u^{(2)})  &\cdots & \pi^{(K)}(b^{(K)},u^{(K)})\end{pmatrix} $$
defined in (11), the conditional law for $\Stacked{Y}$, $\Stacked{B}$ and $\Stacked{U}$ is equal to
\begin{equation*}
\begin{aligned}
&\mathcal{N}_{\mathcal{O}}(\StackedSymbol{\beta}_{\Stacked{E}})^{-1} \Big\lvert{ \small{\left(J_{\Pi_{\Stacked{X}'\Stacked{Y}}} \left(\Stacked{B}, \Stacked{U} \right)\right)}}\Big\rvert \prod_{k=1}^K \rho\left(y^{(k)}; X_{E_k}^{(k)}\beta_{E_k}^{(k)}, \sigma_k^2 \cdot \rm{I}_{n_k} \right) \\ & \;\;\;\;\;\;\;\;\; \cdot \exp\big(-\frac{1}{2}\sum_{k=1}^K (\pi^{(k)}(b^{(k)},u^{(k)}))' {\Omega^{(k)}}^{-1} \pi^{(k)}(b^{(k)},u^{(k)})\big) 1_{\left\{b^{(k)} >\mathbf{0}, \|u^{(k)}\|_{\infty}<1 \text{ for } k \in [K] \right\}} (\Stacked{B}, \Stacked{U}),
\end{aligned}
\end{equation*}
where we have used the equivalence of 
$$\h{\Stacked{E}}(\Stacked{Y}, \Stacked{W}) = \Stacked{E}, \; \h{\Stacked{S}}(\Stacked{Y}, \Stacked{W}) =\Stacked{S}$$
to the following constraints
$$b^{(k)} >\mathbf{0}, \|u^{(k)}\|_{\infty}<1 \text{ for all } k \text{ in } \{1,2,\cdots,K\}.$$
Ignoring constants in $\StackedSymbol{\beta}_{\Stacked{E}}$ in the above expression, we deduce the claim in the Proposition.
\end{proof}

The adjusted likelihood for this conditional prescription, seemingly a first-line recourse for post-selection inference, does not lend itself towards tractable estimating equations. Note that there is no clear method to compute the MLE given the following system of estimating equations:
$$
\dfrac{\partial  }{\partial \beta_{E_k}^{(k)}} \log \mathcal{N}_{\mathcal{O}}(\StackedSymbol{\beta}_{\Stacked{E}})\Big\lvert_{\Mb_{\Stacked{E}}}= \frac{1}{\sigma_k^2}{X^{(k)}_{E_k}}^{\prime} \Big(y^{(k)} - X^{(k)}_{E_k} \beta_{E_k}^{(k)}\Big)\Big\lvert_{\Mb_{\Stacked{E}}} \;\;
\text{ for } k \in [K]
$$
The absence of closed-form expressions for the partial derivatives of the Jacobian $J_{\Pi_{\Stacked{X}'\Stacked{Y}}}(\cdot)$ and the normalizing constant $\mathcal{N}_{\mathcal{O}}(\StackedSymbol{\beta}_\Stacked{E})$ poses a substantial obstacle in solving these equations. An attempt instead to numerically integrate out $Kp+ N$ variables in the derivatives of the normalizing constant will turn out to be quite futile from an angle of computing efficiency.

\subsection{Simulated Experiments}
\label{App:sec3}

\subsubsection{Deviation from Gaussian distribution}

We continue with the simulations in Section 4, and investigate the performance of our methods as the response variables deviate from a Gaussian distribution.
Figure \ref{fig:robustness} shows the distribution of coverage rates for our method when the models errors are drawn from Exponential and Laplace distributions.

\begin{figure}[h!]
  \centering
    \includegraphics[width=\linewidth]{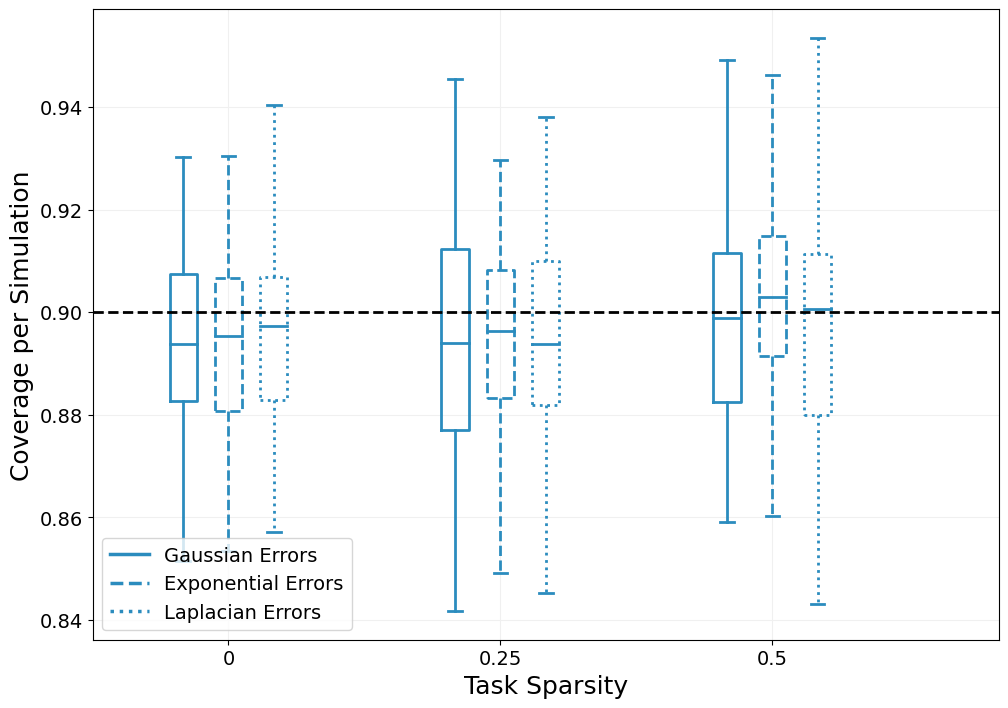}
  \caption{Coverage rates for ``MTL + SI'' as a function of task sparsity $s_T$. 
}
\label{fig:robustness}
\end{figure}

We note that ``MTL + SI'' continues to attain nominal coverage levels on synthetic data matched to the dimensions and estimated sparsity level of the ABCD data even when the response variables deviate from the Gaussian distribution. 
One expects this behavior because the added Gaussian randomization ensures asymptotically-valid selective inference just as does data-splitting, thanks to a Central Limit Theorem proved in \cite{panigrahi2019carving}.

\subsubsection{Varying Sparsity and Dimensions}
In this section, we design three sets of examples to study the effects of (1) varying the heterogeneity in predictors across tasks  by varying task sparsity $s_T$ at a fixed value of global sparsity $s_G$; (2) varying global sparsity $s_G$ at a fixed value of task sparsity $s_T$; and (3) varying the number of predictors $p$, fixing $s_T$ and the total number of active predictors. 
In all of the examples, we follow the generative scheme in Section 4, and set $n=500$, $K=5$, $\rho=0.3$, and $\alpha=0.1$. For scenario (1), we additionally fix $p=100$,  $s_{G}=0.9$, and vary task sparsity $s_{T}$ from $0$ to $0.6$.  The signal-to-noise ratio per task, defined as 
$$\text{SNR} = \frac{{\beta^{(k)}}' \mathcal{T}(\rho) \beta^{(k)}}{n_k \sigma^2},$$
averages around 0.1 across tasks in the first set of experiments. For scenario (2), we fix $p=100$, $s_T= 0.2$, and vary $s_G$ between $0.8$ and $0.95$. The average SNR per task is around $0.5$ for global sparsity of 80 percent and decreases with increasing global sparsity. In both of the first two scenarios, we perform $100$ replications. The optimal tuning parameter for each method is  the value on the $\lambda$-path that yields the lowest average MSE on the validation set across iterations, and the results are reported at that tuning parameter value for all 100 iterates.  For setting (3), we fix $t_{S}=0.2$, set $s_G$ so that the expected number of active predictors per task is 10, and vary $p$ between 100 and 1000. The average SNR-per task increases to around $2.5$ for the highest-dimensional setting to compensate for the added complexity. Due to the computational cost required to tune the model for larger $p$,  we use $10$ auxiliary iterations to choose tuning parameters and then use these parameters for the $100$ iterations we report.
The same metrics, \text{CR}, averaged lengths of confidence intervals, and \text{F1}, are used to evaluate performance.

We start from examining the results over the entire path of the tuning parameter $\lambda$.   Figure \ref{fig:lambda_path} shows the coverage, interval length, and F1 score for a range of values of $\lambda$  for the setting $s_G=0.9$, $s_T=0.2$, and $p=100$.   The distribution of coverage for  the ``naive" method underscores the dangers of ignoring model selection bias, with much higher variability in coverage than the post-selection methods. Reassuringly, all the post-selection methods achieve the nominal coverage rate of  $0.9$, regardless of the choice of tuning parameter and subsequent quality of model selected. The interval lengths along the path of $\lambda$ show that the methods based on randomization consistently produce tighter confidence intervals than the equivalent data splitting approach. 
However, it is not possible to directly compare the accuracy of inference from each method at the same, fixed value of the tuning parameter since the quality of the model selected by each method varies considerably along with $\lambda$.  The F1 score, for example, achieves its best value at different points along the path for different methods.  
The colored dots in the top panel of Figure~\ref{fig:vary_task_sparsity} show the value of $\lambda$ that yielded the lowest mean squared error (MSE) on the holdout data set.  From this point on, we will compare methods using the optimal value of the tuning parameter, chosen by minimizing the MSE on the validation data, for each method.  

\begin{figure}[h!]
  \centering
    \includegraphics[width=\linewidth]{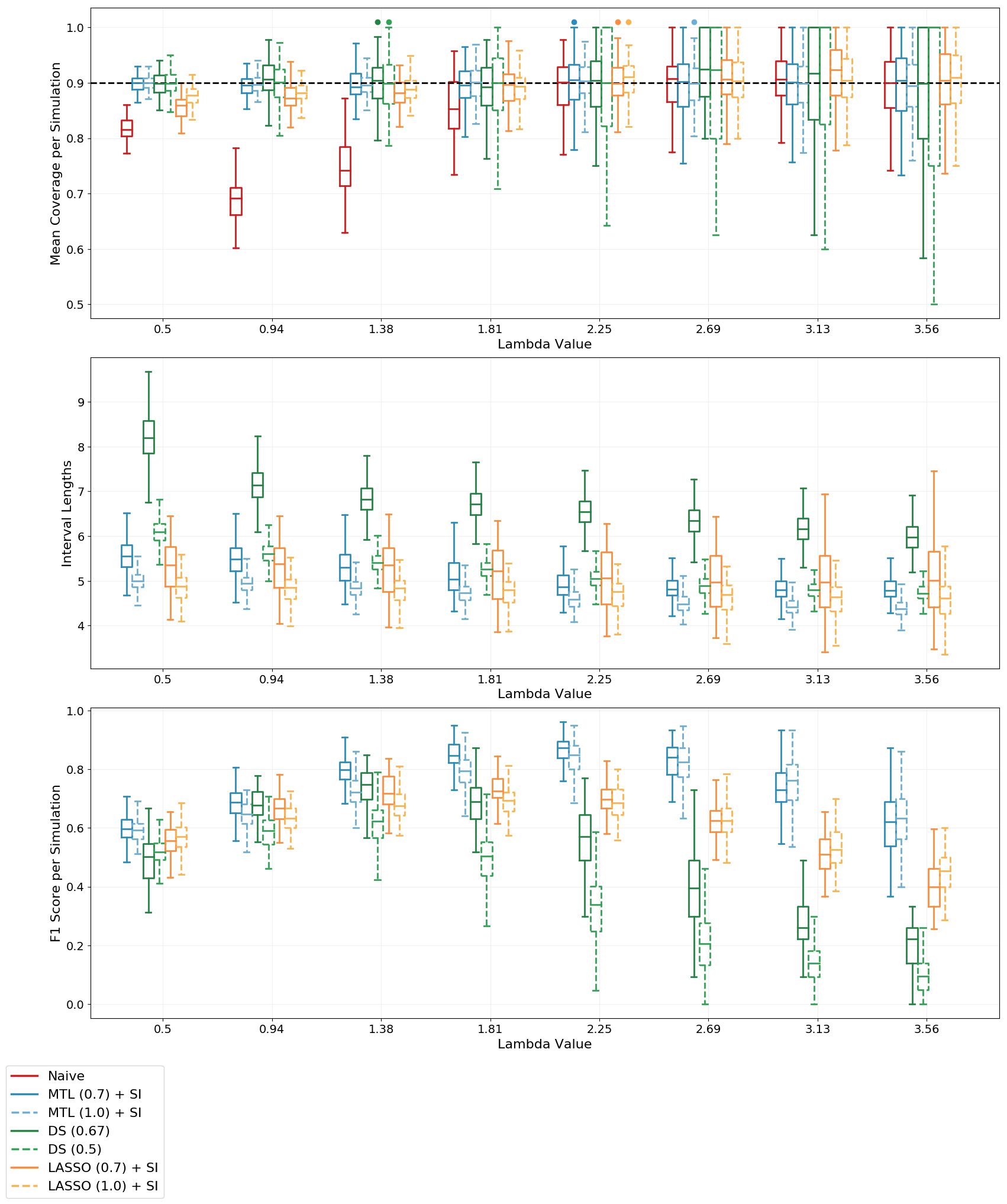}
  \caption{Coverage, length, and accuracy along the path of $\lambda$ for the setting $s_G=0.9$, $s_T=0.2$, and $p=100$. The solid and dashed lines are used to group methods that use a comparable amount of information for model selection; randomization parameters $v=0.70$ and $v=1$ are comparable to data splitting with $s=0.67$ and $s=0.50$, respectively. 
}
  \label{fig:lambda_path}
\end{figure}

Figure \ref{fig:vary_task_sparsity} shows the distribution of coverage, interval length, and the F1 score as we vary the task sparsity $s_T$ from 0 to 0.6, holding $s_G$ and $p$ fixed.    All methods consistently attain nominal or near-nominal coverage, regardless of task sparsity.      Selective inference  (MTL $+$ SI) yields consistently shorter intervals than data splitting  (DS), illustrating the benefit of  using the whole sample rather than splitting.   %The length of `MTL $+$ SI" intervals appears to have lower variance than the length of ``LASSO $+$ SI" intervals.
At all levels of task sparsity, but especially when there is more homogeneity in active predictors across tasks, MTL methods with selective inference also achieve better accuracy than single-task LASSO with selective inference as measured by the F1 score.

\begin{figure}[h!]
  \centering
    \includegraphics[width=\linewidth]{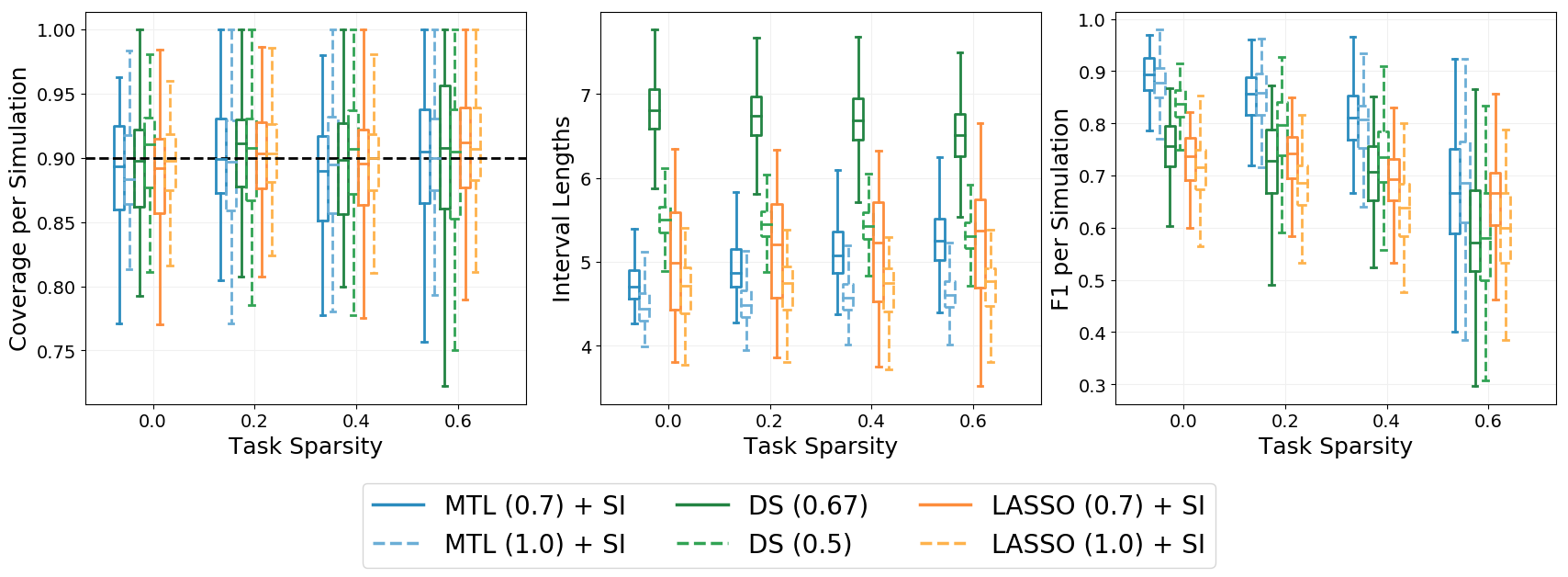}
  \caption{Results for 100 replications as a function of task sparsity $s_T$, fixing global sparsity $s_G = 0.9$ and $p=100$.  Top: coverage; middle: interval length; bottom: F1 score.}
  \label{fig:vary_task_sparsity}
\end{figure}

Figure \ref{fig:vary_global_sparsity}  shows the same patterns continue to hold when we vary global sparsity  while keeping task sparsity fixed.   We only vary global sparsity between 0.8 and 0.95 because the problem becomes very easy for all methods at lower levels of global sparsity.      MTL$+$ SI again outperforms other methods on both interval length and accuracy at all levels of global sparsity. We note, however, that the variance increases for all methods at larger levels of $s_G$.

\begin{figure}[h!]
  \centering
    \includegraphics[width=\linewidth]{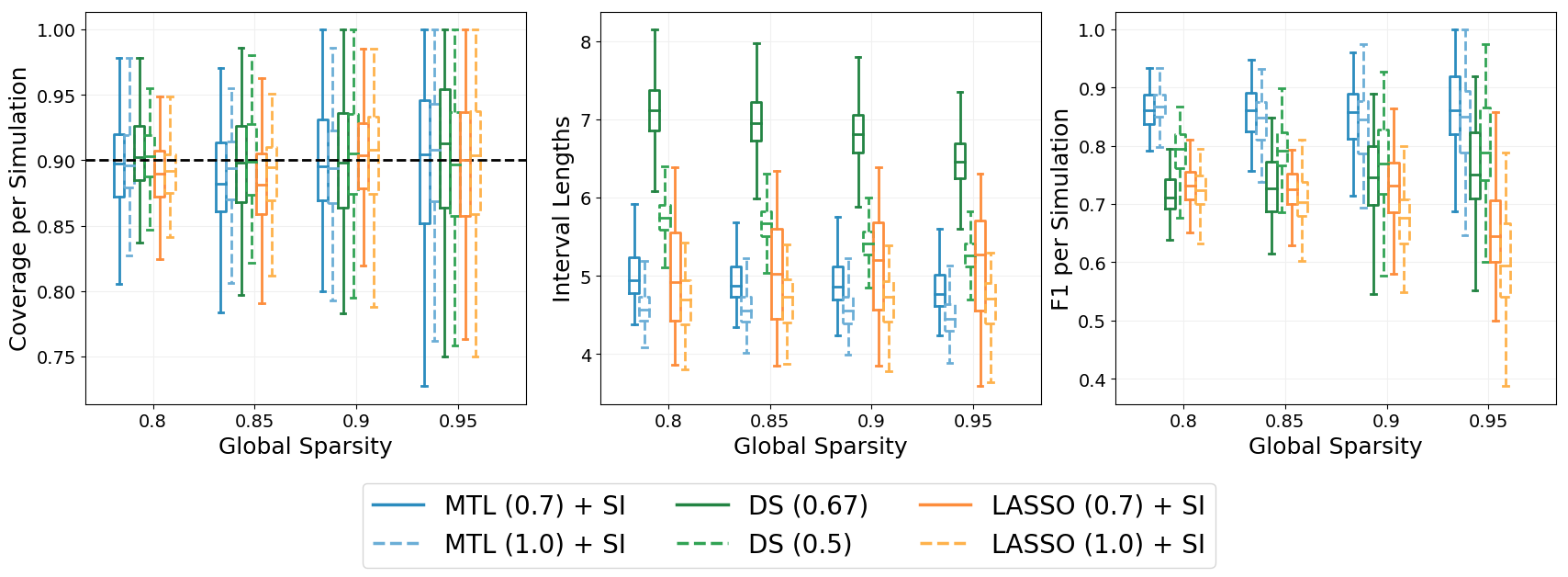}
  \caption{Results for 100 replications as a function of global sparsity $s_G$, fixing task sparsity $s_T = 0.2$ and $p=100$.   Left: coverage; middle: interval length; right: F1 score.   }
  \label{fig:vary_global_sparsity}
\end{figure}

Finally, Figure \ref{fig:vary_p} shows the results as a function of the number of predictors $p$.    To control the difficulty of the problem as $p$ changes, we fix the average number of active predictors per task at 10, rather than specifying the fraction of active predictors.   The average interval length for the selective inference methods grows slightly with $p$,  reflecting the cost of conditioning upon additional information.   Most overall patterns remain the same, with MTL $+$ SI providing the best performance. The accuracy of LASSO $+$ SI  declines faster with $p$ than that of other methods, indicating the growing cost of disregarding shared information.  In summary, our simulation results show that there is variability in performance of all methods, but multi-task learning with selective inference achieves the nominal coverage in all settings and provides the best performance, on average, compared to either data splitting or single-task prediction.   

\begin{figure}[h!]
  \centering
    \includegraphics[width=\linewidth]{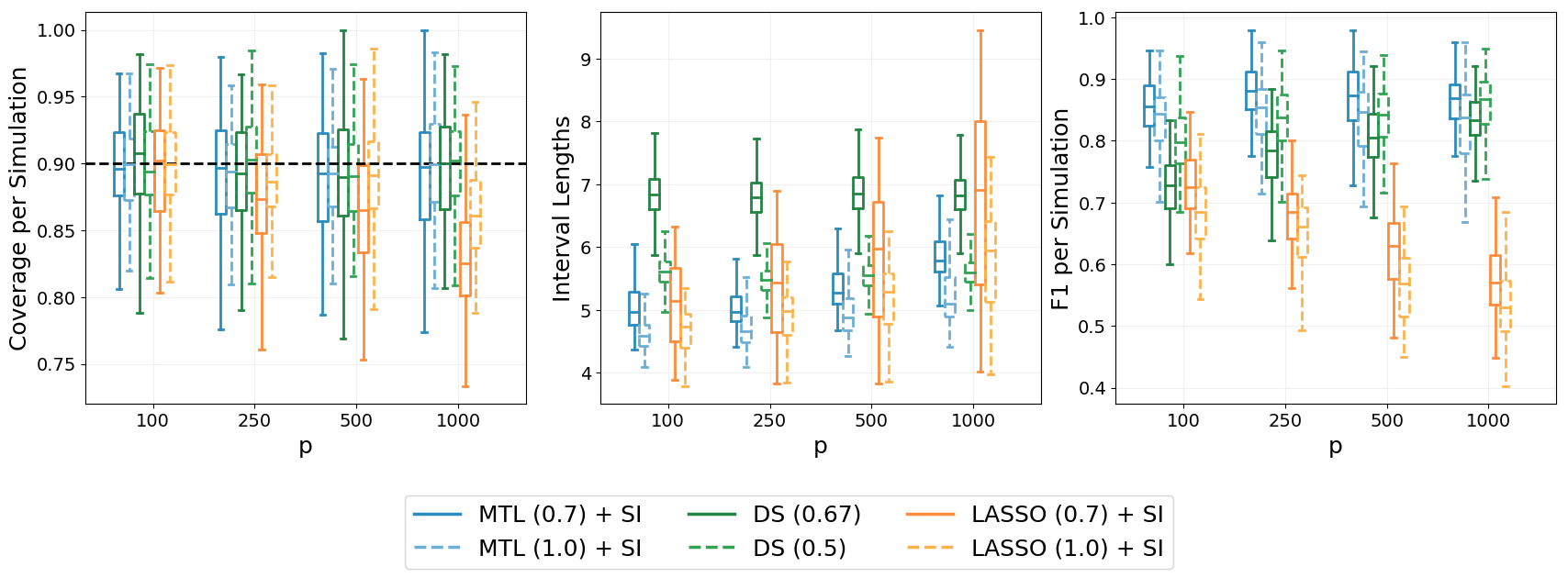}
  \caption{Results for 100 replications as a function of the number of predictors $p$,  fixing task sparsity $s_T = 0$ and an average of 10 active predictors per task. Left: coverage; middle: interval length; right: F1 score.   }
    \label{fig:vary_p}
\end{figure}

\end{document}